%% file: root.tex
\documentclass[orivec]{llncs}

\usepackage[margin=1in]{geometry}
\usepackage{amsmath}
\usepackage{amssymb}
\usepackage{bussproofs}
\usepackage{graphicx}
\usepackage{comment}
\usepackage{ntheorem}

\renewtheorem{theorem}{Theorem}[section]
\renewtheorem{lemma}[theorem]{Lemma}
\renewtheorem{corollary}[theorem]{Corollary}
\renewtheorem{proposition}[theorem]{Proposition}
\renewtheorem{definition}{Definition}[section]

\newcommand{\limp}[0]{\ensuremath{\rightarrow}}

\newcommand{\mimp}[0]{\ensuremath{{-\!*\;}}}
\newcommand{\munit}{\ensuremath{\top^*}}

\newcommand{\simp}[0]{\ensuremath{\triangleright}}

\newcommand{\mand}[0]{\ensuremath{*}}
\newcommand{\lsbbi}[0]{\ensuremath{LS_{BBI}}}
\newcommand{\fvlsbbi}[0]{\ensuremath{FVLS_{BBI}}}

\newcommand{\psl}[0]{\ensuremath{\mbox{PASL}}}
\newcommand{\lspslh}[0]{\ensuremath{LS_{PASL}}}
\newcommand{\ilspslh}[0]{\ensuremath{ILS_{PASL}}}
\newcommand{\ilspslhe}[0]{\ensuremath{ILS_{PASL}2}}

\newcommand{\lsst}[0]{\ensuremath{LS_{ST}}}

\newcommand{\myframe}[3]{\ensuremath{(#1, #2, #3)}}
\newcommand{\mymodel}[4]{\ensuremath{(#1, #2, #3, #4)}}
\newcommand{\mylmodel}[5]{\ensuremath{(#1, #2, #3, #4, #5)}}
\newcommand{\myseq}[3]{\ensuremath{#1; #2 \vdash #3}}

\newcommand{\pfun}{\rightharpoonup}

\def\Fcal{\mathcal{F}}
\def\Gcal{\mathcal{G}}
\def\Mcal{\mathcal{M}}
\def\Ncal{\mathcal{N}}
\def\Lcal{\mathcal{L}}
\def\Pcal{\mathcal{P}}
\def\Scal{\mathcal{S}}

\def\subst{\mathit{subst}}

\def\val{\nu}

\def\Abb{\mathbb{A}}
\def\Ebb{\mathbb{E}}
\def\Ubb{\mathbb{U}}
\def\Sbb{\mathbb{S}}
\def\EMbb{\mathbb{E}\mathbb{M}}
\def\CSbb{\mathbb{C}\mathbb{S}}

\begin{document}

\title{Proof Search for Propositional Abstract Separation Logics via
Labelled Sequents}

\author{
Zh\'e H\'ou\inst{1} \and Ranald Clouston\inst{1} \and Rajeev Gor\'e\inst{1} \and Alwen Tiu\inst{1,2}
}
\institute{
Research School of Computer Science,
The Australian National University
\and
School of Computer Engineering,
Nanyang Technological University 
}

\maketitle

\begin{abstract}
Abstract separation logics are a family of extensions of Hoare logic
for reasoning about programs that mutate memory. These logics are
``abstract'' because they are independent of any particular concrete
memory model. Their assertion languages, called propositional abstract
separation logics, extend the logic of (Boolean) Bunched Implications
(BBI) in various ways.

We develop a modular proof theory for various propositional abstract
separation logics using cut-free labelled sequent calculi. We first
extend the cut-fee labelled sequent calculus for BBI of H\'{o}u et al
to handle Calcagno et al's original logic of separation algebras by
adding sound rules for partial-determinism and cancellativity, while
preserving cut-elimination. We prove the completeness of our calculus
via a sound intermediate calculus that enables us to construct
counter-models from the failure to find a proof. We then capture other
propositional abstract separation logics by adding sound rules for
indivisible unit and disjointness, while maintaining completeness and
cut-elimination. We present a theorem prover based on our labelled
calculus for
these 
logics.
\end{abstract}




\input{intro}

\input{psl}
\input{cut_elim_psl}
\input{complete_interm_psl}
\input{counter_model_psl}
\input{extension_psl}
\input{examples_sl}
\input{sl}

\input{conclusion}
\bibliographystyle{plain}
\bibliography{root}

\appendix
\input{appendix}


\end{document}

%% file: intro.tex
\section{Introduction}

Separation logic (SL)~\cite{reynolds2002} is an extension of Hoare logic for
reasoning about programs that explicitly mutate memory. This is achieved via an
assertion language
that, along with the usual (additive) connectives and predicates for first-order
logic with arithmetic, has the multiplicative connectives \emph{separating
conjunction} $\mand$, its unit $\munit$, and \emph{separating implication}, or
\emph{magic wand}, $\mimp$, from the logic of Bunched Implications
(BI)~\cite{OHearnPym1999}, as well as the \emph{points-to} predicate $\mapsto$.
The additive connectives may be either intuitionistic, as for BI, or classical,
as for the logic of \emph{Boolean} Bunched Implications (BBI). Classical
additives are more expressive as they support reasoning about non-monotonic
commands such as memory deallocation, and assertions such as ``the heap is
empty''~\cite{IshtiaqOHearn01}. In this paper we consider classical additives
only.

The assertion language of SL must provide a notion of inference to support
precondition strengthening and postcondition weakening, yet little such proof
theory exists, despite its link with the proof-theoretically motivated
BI. Instead, inference must proceed via reasoning directly about the concrete semantics of
\emph{heaps}, or finite partial functions from addresses to values. A heap
satisfies $P\mand Q$ iff it can be partitioned into heaps satisfying $P$ and $Q$
respectively; it satisfies $\munit$ iff it is empty; it satisfies $P\mimp Q$ iff
any extension with a heap that satisfies $P$ must then satisfy $Q$; and it
satisfies $E\mapsto E'$ iff it is a singleton map sending the address specified
by the expression $E$ to the value specified by the expression $E'$. Such
concrete semantics are appropriate for proving the correctness of a specific
program in a specific environment, but mean that if a different notion of memory
(or more generally, resource) is required then a new logic is also required.

Calcagno et al's Abstract Separation Logic (ASL)~\cite{calcagno2007} introduced
the abstract semantics of partial cancellative monoids, or \emph{separation
algebras}, to unify notions of resources for heaps, heaps with permissions,
Petri nets, and other examples. These semantics allow interpretation of $\mand$,
$\munit$ and $\mimp$, although the latter is not considered by Calcagno et al.
However $\mapsto$ has no meaning in separation algebras in general, and is
therefore not a first class citizen of ASL; it may be introduced as a predicate
only if an appropriate concrete separation algebra is fixed. Calcagno et al
do not consider proof theory for their assertion language, whose propositional
fragment we call Propositional Abstract Separation Logic (PASL), but separation
algebras are a restriction of \emph{non-deterministic monoids}, which are known
to give sound and complete semantics for BBI~\cite{GalmicheLarcheyWendling2006}.
In this sense PASL is a refinement of BBI, differing only by the addition of the
semantic properties of \emph{partial-determinism} and \emph{cancellativity}.

This link between BBI and PASL semantics raises the question of whether existing
proof theory for BBI can be extended to give a sound and cut-free complete proof
system for PASL; we answer this question in the affirmative by
extending the labelled sequent calculus $LS_{BBI}$ of H\'{o}u et
al~\cite{hou2013} by adding explicit rules for partial-determinism and cancellativity.
The completeness of $LS_{BBI}$ was demonstrated via the Hilbert axiomatisation of
BBI, but this avenue is not open to us as partial-determinism and cancellativity are not
axiomatisable in BBI~\cite{brotherston2013}; instead completeness follows via a
\emph{counter-model construction} procedure. A novelty of our counter-model construction
is that it can be modularly extended to handle extensions and sublogics of PASL.

We have also implemented proof search
using our calculus (although no decision procedure for PASL is
possible~\cite{brotherston2010}). To our knowledge this is the first proof to be
presented of the cut-free completeness of a calculus for PASL%
\footnote{Larchey-Wendling~\cite{wendling2012} claims that
the tableaux 
for BBI with partial-determinism in \cite{wendling2009} can be extended to cover cancellativity, 
but the ``rather involved'' proof  has not appeared yet.
}%
, and our implementation is the first automated theorem prover for PASL. 

Just as we have a family of separation logics, ranging across different
concrete semantics, we now also have a family of abstract separation logics for
different abstract semantics. These abstract semantics are often expressed as
extensions of the usual notion of separation algebra; most notably Dockins et
al~\cite{dockins2009} suggested the additional properties of positivity (here
called \emph{indivisible unit}), \emph{disjointness}, \emph{cross-split}, and
\emph{splittability}%
\footnote{Dockins et al~\cite{dockins2009} also suggest generalising
separation algebras to have a \emph{set} of units; it is an easy corollary
of~\cite[Lem. 3.11]{brotherston2013} that single-unit and multiple-unit
separation algebras satisfy the same set of formulae.}%
. Conversely, the abstract semantics for Fictional Separation
Logic~\cite{JensenBirkedal2012} generalise separation algebras by dropping
cancellativity. Hence there is demand for a modular approach to proof theory
and proof search for propositional abstract separation logics. Labelled sequent
calculi, with their explicitly semantics-based rules, provide good support for
this modularity, as rules for the various properties can be added and removed
as required. We investigate which properties can be combined without sacrificing
our cut-free completeness result.

While we work with abstract models of separation logics, the reasoning
principles
behind our proof-theoretic methods should be applicable to
concrete models also, so
we investigate as further work how concrete predicates such as $\mapsto$
might be integrated into our approach.
Proof search strategies that come out of our
proof-theoretic analysis could also potentially be applied to guide proof search in
various 
encodings
of separation logics~\cite{Appel2006,Tuch2007,McCreight2009}
in proof assistants, e.g., they can guide the constructions of proof tactics needed to automate the
reasoning tasks in those embeddings.

The remainder of this paper is structured as follows. Section 2 introduces propositional abstract separation logic based on separation algebra semantics, and gives the labelled sequent calculus for this logic. Fundamental results such as soundness and cut-elimination are also shown. Section 3 proves the completeness of the labelled calculus by counter-model construction. Section 4 discusses extensions of the labelled calculus with desirable properties in separation theory.
Implementation and experiment are shown in Section 5, followed by Section 6 where a preliminary future work on how concrete predicates such as $\mapsto$ might be integrated into our approach is outlined. Finally, Section 7 discusses related work.

%% file: psl.tex
\section{The labelled sequent calculus for $\psl$}

In this section we define the \emph{separation algebra} semantics of
Calcagno et al~\cite{calcagno2007} for Propositional Abstract Separation Logic
(PASL), and present the labelled sequent calculus $\lspslh$ for this logic,
extending the calculus $LS_{BBI}$ for BBI of H\'{o}u et al~\cite{hou2013}
with partial-determinism and cancellativity.
Soundness and cut-elimination are then demonstrated for $\lspslh$.

%

\subsection{Propositional abstract separation logic}
\label{subsec:pasl}
The formulae of $\psl$ are defined inductively as follows, where $p$ ranges over
some set $Var$ of propositional variables:
\begin{align*}
A ::= & \ p\mid \top\mid \bot\mid\lnot A\mid A\lor A\mid A\land A\mid A\limp A\mid \top^*\mid A\mand A\mid A\mimp A
\end{align*}
PASL-formulae are interpreted with respect to the following semantics:

\begin{definition}\label{dfn:sepalg}
A \emph{separation algebra}, or partial cancellative commutative monoid, is a
triple $(H,\circ,\epsilon)$ where $H$ is a non-empty set, $\circ$ is a partial
binary function $H\times H\pfun H$ written infix, and $\epsilon\in H$,
satisfying the following conditions, where `$=$' is interpreted as `both sides
undefined, or both sides defined and equal':
\begin{description}
\item[identity:] $\forall h\in H.\, h\circ\epsilon = h$
\item[commutativity:] $\forall h_1,h_2\in H.\, h_1\circ h_2 = h_2\circ h_1$
\item[associativity:] $\forall h_1,h_2,h_3\in H.\, h_1\circ(h_2\circ h_3) = (h_1\circ h_2)\circ h_3$
\item[cancellativity:] $\forall h_1,h_2,h_3,h_4\in H.$ if $h_1\circ h_2 = h_3$ and $h_1\circ h_4 = h_3$ then $h_2 = h_4$
\end{description}
\end{definition}

Note that \textit{partial-determinism} of the monoid is assumed since 
$\circ$ is a partial function: 
for any $h_1,h_2,h_3,h_4\in H$, if $h_1\circ h_2 = h_3$ and $h_1\circ h_2 = h_4$ then $h_3 = h_4$. The paradigmatic example of a separation algebra is the set of \emph{heaps};
here $\circ$ is the combination of two heaps with disjoint domain, and
$\epsilon$ is the empty heap.

In this paper we prefer to express PASL semantics in the style of \emph{ternary
relations}, to maintain consistency with the earlier work of H\'{o}u et
al on BBI~\cite{hou2013}; it is easy to see that the definition below is a
trivial notational variant of Def.~\ref{dfn:sepalg}.

\begin{definition}
A \emph{PASL Kripke relational frame} is a triple $(H, R, \epsilon)$, where $H$
is a non-empty set of \emph{worlds}, $R \subseteq H \times H \times H$, and
$\epsilon \in H$, satisfying the following conditions for all $h_1,h_2,h_3,h_4,
h_5$ in $H$:
\begin{description}
\item[identity:] $R(h_1,\epsilon, h_2) \Leftrightarrow$ $h_1 = h_2$
\item[commutativity:] $R(h_1,h_2, h_3) \Leftrightarrow$ $R(h_2,h_1,
  h_3)$
\item[associativity:] $(R(h_1,h_5, h_4) \& R(h_2,h_3,h_5)) \Rightarrow$
 $\exists h_6.(R(h_6,h_3, h_4) \& R(h_1,h_2, h_6))$
\item[cancellativity:] $(R(h_1,h_2,h_3) \& R(h_1,h_4,h_3)) \Rightarrow h_2 = h_4$
\item[partial-determinism:] $(R(h_1,h_2,h_3) \& R(h_1,h_2, h_4)) \Rightarrow h_3 = h_4$.
\end{description}
\end{definition}

A \emph{PASL Kripke relational model} is a tuple
$\mymodel{H}{R}{\epsilon}{\val}$ of a PASL Kripke relational frame $(H,R,
\epsilon)$ and a {\em valuation} function $\val : Var \rightarrow \Pcal(H)$
(where $\Pcal(H)$ is the power set of $H$).  The forcing relation $\Vdash$
between a model  $\Mcal = \mymodel{H}{R}{\epsilon}{\val}$
and a formula is defined in Table~\ref{tab:psl_semantics}, where we write
$\Mcal, h \not \Vdash A$ for the negation of $\Mcal, h \Vdash A$. Given a model
$\Mcal = \mymodel{H}{R}{\epsilon}{\val}$, a formula is \emph{true at (world)
$h$} iff $\Mcal,h \Vdash A$. The formula $A$ is {\em valid} iff it is true at
all worlds of all models.
\begin{table*}
\vspace{-0.5cm}
\centering
\begin{tabular}{l@{\extracolsep{1cm}}l}
\begin{tabular}{lcl}
$\Mcal, h \Vdash p$ & iff & $p \in Var$ and $h \in v(p)$
\\
$\Mcal, h \Vdash A\land B$ & iff & $\Mcal, h \Vdash A$ and $\Mcal, h \Vdash B$
\\
$\Mcal, h \Vdash A\limp B$ & iff & $\Mcal, h\not \Vdash A$ or $\Mcal, h\Vdash B$\\
$\Mcal, h \Vdash A\lor B$ & iff & $\Mcal, h \Vdash A$ or $\Mcal, h \Vdash B$
\\
\end{tabular}
&
\begin{tabular}{lcl}
$\Mcal, h \Vdash \top^*$ & iff & $h = \epsilon$\\
$\Mcal, h \Vdash \top$ & iff & always
\\
$\Mcal, h \Vdash \bot$ & iff & never 
\\
$\Mcal, h \Vdash \lnot A$ & iff & $\Mcal, h \not \Vdash A$
\end{tabular}\\
\multicolumn{2}{l}{$\;\;\Mcal, h \Vdash A\mand B$ iff 
$\exists h_1,h_2.(R(h_1,h_2, h)$ and $\Mcal, h_1\Vdash A$ and $\Mcal, h_2 \Vdash B)$}\\
\multicolumn{2}{l}{$\;\;\Mcal, h \Vdash A \mimp B$ iff 
$\forall h_1,h_2.((R(h,h_1,h_2)$ and $\Mcal, h_1 \Vdash A)$ implies $\Mcal, h_2 \Vdash B)$}\\
\end{tabular}\ \\[3px]
\caption{Semantics of PASL, where $\Mcal = \mymodel{H}{R}{\epsilon}{\val}.$}
\label{tab:psl_semantics}
\vspace{-1cm}
\end{table*}

\subsection{The labelled sequent calculus $\lspslh$}
\label{subsec:lspsl}


Let $LVar$ be an infinite set of {\em label variables}, and let the set $\Lcal$ of \emph{labels} be $LVar\cup\{\epsilon\}$, where $\epsilon$ is a label constant
not in $LVar$; here we overload the notation for the identity world in the
semantics. Labels will be denoted by lower-case letters such as $a,b,x,y,z$.
A {\em labelled formula} is a pair $a : A$ of a label $a$ and formula $A$.
As usual in a labelled sequent calculus
one needs to incorporate Kripke relations explicitly into the sequents. This 
is achieved via the syntactic notion of {\em relational atoms}, which have
the form $(a, b \simp c)$, where $a,b,c$ are labels. 
A \emph{sequent} takes the form
$$\myseq{\Gcal}{\Gamma}{\Delta}$$
where $\Gcal$ is a set of relational atoms, 
and $\Gamma$ and $\Delta$ are multisets of labelled formulae. 
Then, $\Gamma;A$ is the multiset union of $\Gamma$ and $\{A\}$.

As the interpretation of the logical connectives of $\psl$ are the same as those
for BBI, we may obtain a labelled sequent calculus for $\psl$, called
$\lspslh$,
by adding the rules $P$ (partial-determinism) and $C$ (cancellativity) to
$\lsbbi$~\cite{hou2013}. The rules for $\lspslh$ are presented in
Fig.~\ref{fig:LS_PSL_H}, where $p$ is a propositional variable, $A,B$ are
formulae,  and $w,x,y,z\in\Lcal$. Note that some rules use
\emph{label substitutions}. We write $\Gcal[y/x]$ (resp. $\Gamma[y/x]$) for the 
(multi)set of relational atoms (resp. labelled formulae) for which 
the label variable $x$ has been uniformly replaced by the label $y$.
In each rule, the formula (resp. relational atom) shown explicitly 
in the conclusion is called the \textit{principal} formula (resp. relational
atom). A rule with no premise is called a \textit{zero-premise} rule.
Note that the $\limp L$ rule is the classical implication left rule. 



\begin{figure*}[!t]
\footnotesize
\centering
\begin{tabular}{cc}
\multicolumn{2}{l}{\textbf{Identity and Cut:}}\\[10px]
\AxiomC{$$}
\RightLabel{\tiny $id$}
\UnaryInfC{$\myseq{\Gcal}{\Gamma;w:p}{w:p;\Delta}$}
\DisplayProof
&
\AxiomC{$\myseq{\Gcal}{\Gamma}{x:A;\Delta}$}
\AxiomC{$\myseq{\Gcal'}{\Gamma';x:A}{\Delta'}$}
\RightLabel{\tiny $cut$}
\BinaryInfC{$\myseq{\Gcal;\Gcal'}{\Gamma;\Gamma'}{\Delta;\Delta'}$}
\DisplayProof
\\[15px]
\multicolumn{2}{l}{\textbf{Logical Rules:}}\\[10px]
\AxiomC{$$}
\RightLabel{\tiny $\bot L$}
\UnaryInfC{$\myseq{\Gcal}{\Gamma; w:\bot}{\Delta}$}
\DisplayProof
$\qquad$
\AxiomC{$\myseq{(\epsilon,w\simp \epsilon);\Gcal}{\Gamma}{\Delta}$}
\RightLabel{\tiny $\top^* L$}
\UnaryInfC{$\myseq{\Gcal}{\Gamma;w:\top^*}{\Delta}$}
\DisplayProof
&
\AxiomC{$$}
\RightLabel{\tiny $\top R$}
\UnaryInfC{$\myseq{\Gcal}{\Gamma}{w:\top;\Delta}$}
\DisplayProof
$\qquad$
\AxiomC{$$}
\RightLabel{\tiny $\top^* R$}
\UnaryInfC{$\myseq{\Gcal}{\Gamma}{\epsilon:\top^*;\Delta}$}
\DisplayProof\\[15px]
\AxiomC{$\myseq{\Gcal}{\Gamma;w:A;w:B}{\Delta}$}
\RightLabel{\tiny $\land L$}
\UnaryInfC{$\myseq{\Gcal}{\Gamma;w:A\land B}{\Delta}$}
\DisplayProof
&
\AxiomC{$\myseq{\Gcal}{\Gamma}{w:A;\Delta}$}
\AxiomC{$\myseq{\Gcal}{\Gamma}{w:B;\Delta}$}
\RightLabel{\tiny $\land R$}
\BinaryInfC{$\myseq{\Gcal}{\Gamma}{w:A\land B;\Delta}$}
\DisplayProof\\[15px]
\AxiomC{$\myseq{\Gcal}{\Gamma}{w:A;\Delta}$}
\AxiomC{$\myseq{\Gcal}{\Gamma;w:B}{\Delta}$}
\RightLabel{\tiny $\limp L$}
\BinaryInfC{$\myseq{\Gcal}{\Gamma;w:A\limp B}{\Delta}$}
\DisplayProof
&
\AxiomC{$\myseq{\Gcal}{\Gamma;w:A}{w:B;\Delta}$}
\RightLabel{\tiny $\limp R$}
\UnaryInfC{$\myseq{\Gcal}{\Gamma}{w:A\limp B; \Delta}$}
\DisplayProof\\[15px]
\AxiomC{$\myseq{(x,y \simp z);\Gcal}{\Gamma;x:A;y:B}{\Delta}$}
\RightLabel{\tiny $\mand L$}
\UnaryInfC{$\myseq{\Gcal}{\Gamma;z:A\mand B}{\Delta}$}
\DisplayProof
&
\AxiomC{$\myseq{(x,z \simp y);\Gcal}{\Gamma;x:A}{y:B;\Delta}$}
\RightLabel{\tiny $\mimp R$}
\UnaryInfC{$\myseq{\Gcal}{\Gamma}{z:A\mimp B;\Delta}$}
\DisplayProof\\[15px]
\multicolumn{2}{c}{
\AxiomC{$\myseq{(x,y \simp z);\Gcal}{\Gamma}{x:A;z:A\mand B;\Delta}$}
\AxiomC{$\myseq{(x,y \simp z);\Gcal}{\Gamma}{y:B;z:A\mand B;\Delta}$}
\RightLabel{\tiny $\mand R$}
\BinaryInfC{$\myseq{(x,y \simp z);\Gcal}{\Gamma}{z:A\mand B;\Delta}$}
\DisplayProof
}\\[15px]
\multicolumn{2}{c}{
\AxiomC{$\myseq{(x,y \simp z);\Gcal}{\Gamma;y:A\mimp B}{x:A;\Delta}$}
\AxiomC{$\myseq{(x,y \simp z);\Gcal}{\Gamma;y:A\mimp B; z:B}{\Delta}$}
\RightLabel{\tiny $\mimp L$}
\BinaryInfC{$\myseq{(x,y \simp z);\Gcal}{\Gamma;y:A\mimp B}{\Delta}$}
\DisplayProof
}
\\[15px]
\multicolumn{2}{l}{\textbf{Structural Rules:}}\\[15px]
\AxiomC{$\myseq{(y,x \simp z);(x,y \simp z);\Gcal}{\Gamma}{\Delta}$}
\RightLabel{\tiny $E$}
\UnaryInfC{$\myseq{(x,y \simp z);\Gcal}{\Gamma}{\Delta}$}
\DisplayProof
&
\AxiomC{$\myseq{(u,w \simp z);(y,v \simp w);(x,y \simp z);(u,v \simp x);\Gcal}{\Gamma}{\Delta}$}
\RightLabel{\tiny $A$}
\UnaryInfC{$\myseq{(x,y \simp z);(u,v \simp x);\Gcal}{\Gamma}{\Delta}$}
\DisplayProof\\[15px]
\AxiomC{$\myseq{(x,\epsilon \simp x);\Gcal}{\Gamma}{\Delta}$}
\RightLabel{\tiny $U$}
\UnaryInfC{$\myseq{\Gcal}{\Gamma}{\Delta}$}
\DisplayProof
&
\AxiomC{$\myseq{(x,w\simp x);(y,y\simp w);(x,y\simp x);\Gcal}{\Gamma}{\Delta}$}
\RightLabel{\tiny $A_C$}
\UnaryInfC{$\myseq{(x,y\simp x);\Gcal}{\Gamma}{\Delta}$}
\DisplayProof
\\[15px]
\AxiomC{$\myseq{(\epsilon,w'\simp w');\Gcal[w'/w]}{\Gamma[w'/w]}{\Delta[w'/w]}$}
\RightLabel{\tiny $Eq_1$}
\UnaryInfC{$\myseq{(\epsilon,w\simp w');\Gcal}{\Gamma}{\Delta}$}
\DisplayProof
&
\AxiomC{$\myseq{(\epsilon, w' \simp w');\Gcal[w'/w]}{\Gamma[w'/w]}{\Delta[w'/w]}$}
\RightLabel{\tiny $Eq_2$}
\UnaryInfC{$\myseq{(\epsilon,w'\simp w);\Gcal}{\Gamma}{\Delta}$}
\DisplayProof\\[15px]
\AxiomC{$\myseq{(x,y\simp z);\Gcal[z/w]}{\Gamma[z/w]}{\Delta[z/w]}$}
\RightLabel{\tiny $P$}
\UnaryInfC{$\myseq{(x,y\simp z);(x,y\simp w);\Gcal}{\Gamma}{\Delta}$}
\DisplayProof
&
\AxiomC{$\myseq{(x,y\simp z);\Gcal[y/w]}{\Gamma[y/w]}{\Delta[y/w]}$}
\RightLabel{\tiny $C$}
\UnaryInfC{$\myseq{(x,y\simp z);(x,w\simp z);\Gcal}{\Gamma}{\Delta}$}
\DisplayProof\\[15px]
\multicolumn{2}{l}{\textbf{Side conditions:}} \\
\multicolumn{2}{l}{Only label variables (not $\epsilon$) may be substituted for.}\\ 
\multicolumn{2}{l}{In $\mand L$ and $\mimp R$, the labels $x$ and $y$
do not occur in the conclusion. }\\
\multicolumn{2}{l}{In the rules $A,A_C$, the label $w$ does not occur in the conclusion.}
\end{tabular}
\caption{The labelled sequent calculus $\lspslh$ for Propositional Abstract Separation Logic.} 
\label{fig:LS_PSL_H}
\vspace{-0.28cm}
\end{figure*}

A function $\rho: \Lcal \rightarrow H$ from
labels to worlds is a {\em label mapping} 
iff it satisfies $\rho(\epsilon) =  \epsilon$, mapping the
label constant $\epsilon$ to the identity world of $H$. Intuitively, a labelled formula $a:A$ means that formula $A$ is true in world $\rho(a)$.
Thus we define an \textit{extended $\psl$ Kripke relational model} $(H,R,\epsilon,\val,\rho)$ as a model equipped with a label mapping.

\begin{definition}[Sequent Falsifiability]
A sequent $\myseq{\Gcal}{\Gamma}{\Delta}$ is \emph{falsifiable in an extended model}
$\Mcal = \mylmodel{H}{R}{\epsilon}{\val}{\rho}$ if for every $x : A \in \Gamma$, $(a,b\simp c) \in \Gcal$, and for every
$y : B \in \Delta$,
we have $(\Mcal,\rho(x) \Vdash A)$, 
$R(\rho(a),\rho(b),\rho(c))$ and $(\Mcal,\rho(y) \not \Vdash B).$
It is {\em falsifiable} if it is falsifiable in some extended model. 
\end{definition}

To show that a formula $A$ is \emph{valid} in $\lspslh$, we prove $\vdash w:A$ for an
arbitrary label $w$. An example derivation is given in Fig.~\ref{fig:ex}.
\begin{figure*}[!t]
\centering
\AxiomC{}
\RightLabel{\tiny $\top^* R$}
\UnaryInfC{$\myseq{(\epsilon,a\simp a);(a,\epsilon\simp a)}{a:A}{\epsilon:\top^*} $}
\AxiomC{}
\RightLabel{\tiny $id$}
\UnaryInfC{$\myseq{(\epsilon,a\simp a);(a,\epsilon\simp a)}{a:A}{a:A}$}
\RightLabel{\tiny $\mand R$}
\BinaryInfC{$\myseq{(\epsilon,a\simp a);(a,\epsilon\simp a)}{a:A}{ a:\top^*  \mand   A}$}
\RightLabel{\tiny $E$}
\UnaryInfC{$\myseq{(a,\epsilon\simp a)}{a:A}{ a:\top^*  \mand   A}$}
\RightLabel{\tiny $U$}
\UnaryInfC{$\myseq{~}{a:A}{a:\top^*  \mand   A}$}
\RightLabel{\tiny $\limp R$}
\UnaryInfC{$ \myseq{~}{~}{ a:A \limp   (\top^*  \mand   A)}$}
\DisplayProof
\caption{An example derivation in $\lsbbi$.}
\label{fig:ex}
\end{figure*}

\begin{theorem}[Soundness]
For any formula $A$, and
for an arbitrary label $w$,
if the labelled sequent 
$\vdash w:A$ is derivable in $\lspslh$
then
$A$ is valid.
\end{theorem}

\begin{proof}
  We prove that the rules of $\lspslh$ preserve falsifiability
  upwards.
The proof is straightforward so we omit the details; but refer the interested reader to a similar 
  proof for $LS_{BBI}$~\cite{hou2013}.\qed
\end{proof}



%% file: cut_elim_psl.tex
\subsection{Cut-elimination}
\label{subsec:cut-elim}

The only differences between $\lspslh$ and $LS_{BBI}$~\cite{hou2013} are the
additions of the structural rules $P$ and $C$, so we may prove cut-elimination
by the same route, which in turn follows from the usual cut-elimination
procedure for labelled sequent calculi for modal logics~\cite{negri2001}. 
We therefore delay the proofs to Appendix~\ref{sec:app1}, and simply list the necessary lemmas. 
In the sequel we use $ht(\Pi)$ to denote the height of the derivation $\Pi$. 

\begin{lemma}[Substitution]
\label{lem:subs}
If
$\Pi$ is an $\lspslh$ derivation for the sequent $\myseq{\Gcal}{\Gamma}{\Delta}$
then there is an $\lspslh$ derivation $\Pi'$ of 
the sequent $\myseq{\Gcal[y/x]}{\Gamma[y/x]}{\Delta[y/x]}$ 
such that $ht(\Pi') \leq ht(\Pi)$.
\end{lemma}

\begin{lemma}[Admissibility of weakening]
If $\myseq{\Gcal}{\Gamma}{\Delta}$ is derivable in $\lspslh$, then for
all structures $\Gcal,\Gamma'$ and $\Delta'$, the sequent
$\myseq{\Gcal;\Gcal'}{\Gamma;\Gamma'}{\Delta;\Delta'}$ is derivable with
the same height in $\lspslh$.
\end{lemma}

\begin{lemma}[Invertibility]
\label{lem:invert}
If $\Pi$ is a cut-free $\lspslh$ derivation of the conclusion of a rule,
then there is a cut-free $\lspslh$ derivation for each premise, with
height at most $ht(\Pi).$
\end{lemma}

\begin{lemma}[Admissibility of contraction]
If $\myseq{\Gcal;\Gcal}{\Gamma;\Gamma}{\Delta;\Delta}$ is derivable in $\lspslh$, 
then $\myseq{\Gcal}{\Gamma}{\Delta}$ is derivable with the same height in $\lspslh$.
\end{lemma} 

\begin{theorem}[Cut-elimination]
If $\myseq{\Gcal}{\Gamma}{\Delta}$ is derivable in $\lspslh$ then it is derivable without using the $cut$ rule.
\end{theorem}
\begin{proof}
The proof follows the same structure as that for $\lsbbi$, utilising
the lemmas above. 
The additional cases we need to consider are those involving the rules 
$P$ and $C$;
their treatment is similar to that for $Eq_1$ in the proof for $\lsbbi$~\cite{hou2013}.\qed
\end{proof}

Since partial-determinism and cancellativity are not axiomatisable in
BBI~\cite{brotherston2013}, cut-elimination does not immediately yield the
completeness of $\lspslh$; we prove completeness of our calculus in the next section.



%% file: complete_interm_psl.tex
\section{Completeness of $\lspslh$}
\label{sec:complete_lspslh}
We prove the completeness of $\lspslh$ with respect to the Kripke relational
semantics by a counter-model construction. A standard way to construct a counter-model
for an unprovable sequent is to show that it can be saturated 
by repeatedly applying all applicable inference rules
to reach a limit sequent where a counter-model can be constructed. 
In adopting such a counter-model construction strategy to $\lspslh$ we
encounter difficulty in formulating the saturation conditions for rules
involving label substitutions.
We therefore adopt the approach of H\'{o}u et al~\cite{hou2013}, using an
intermediate system without explicit use of label substitutions, but
where equivalences between labels are captured via an entailment $\vdash_E$.

\subsection{The intermediate system $\ilspslh$}
\label{subsec:interm}
We introduce an intermediate system where rules with substitutions ($Eq_1$, $Eq_2$, $P$, $C$) are isolated  into an equivalence entailment $\vdash_E$, so that the resultant calculus does
not involve substitutions.

Let $r$ be an instance of a structural rule in which the substitution
used is $\theta$: this is the identity
substitution except when $r$ is $Eq_1$, $Eq_2$, $P$ or $C$.  We can view $r$
(upwards) as a function that takes a set of relational atoms (in the
conclusion of the rule) and outputs another set (in the premise). We
write $r(\Gcal,\theta)$
for the output relational atoms of an instance of $r$ with
substitution $\theta$ and with conclusion containing $\Gcal$.  Let
$\sigma$ be a sequence of instances of structural rules
$[r_1(\Gcal_1,\theta_1);\cdots;r_n(\Gcal_n,\theta_n)]$.  Given a set
of relational atoms $\Gcal$, the result of the (backward) application
of $\sigma$ to $\Gcal$, denoted by $\Scal(\Gcal, \sigma)$, is defined
as below, where $^\smallfrown$ is used for sequence concatenation:
\begin{eqnarray*}
\Scal(\Gcal, \sigma)
 =
 \left\{
     \begin{array}{ll}
      \Gcal & \mbox{if } \sigma = [~] \\
      \Scal(\Gcal \theta \cup r(\Gcal',\theta), \sigma') 
            & \mbox{if }
             \Gcal' \subseteq \Gcal \mbox{ and }\\ 
            & \sigma = [r(\Gcal',\theta)]^\smallfrown\sigma'
     \\
       \mbox{undefined} & \mbox{ otherwise}
     \end{array}
 \right.
\end{eqnarray*}

Given $\sigma = [r_1(\Gcal_1,\theta_1);\cdots;r_n(\Gcal_n,\theta_n)]$, let
$\subst(\sigma)$ be the composite substitution $\theta_1 \circ
\cdots \circ \theta_n$, where $t(\theta_1\circ \theta_2)$ means $(t\theta_1)\theta_2$. 
We write $s \equiv t$ to mean that $s$ and $t$ are syntactically equal. 

\begin{definition}[Equivalence entailment]
\label{dfn:equiv-entail}
Let $\Gcal$ be a set of relational atoms. The entailment 
$\Gcal\vdash_E (a = b)$ holds iff there exists a sequence $\sigma$ of
$Eq_1,Eq_2,P,C$ applications s.t. $\Scal(\Gcal,\sigma)$ is defined,
and $a\theta \equiv b\theta$, where $\theta = subst(\sigma)$.
\end{definition}

Since substitution is no longer in the calculus, some inference rules
that involve matching two equal labels need to be changed. We define the
intermediate system $\ilspslh$ as $\lspslh$ minus 
$\{Eq_1,Eq_2,P,C\}$, 
with certain rules changed following Fig.~\ref{fig:ILSPSLH}. Note that the
equivalence entailment $\vdash_E$ is not a premise, but rather a condition of
the rules.

\begin{figure*}[!ht]
\footnotesize
\centering
\begin{tabular}{cc}
\AxiomC{$\mathcal{G} \vdash_E (w_1 = w_2)$}
\RightLabel{\tiny $id$}
\UnaryInfC{$\myseq{\mathcal{G}}{\Gamma;w_1:p}{w_2:p;\Delta}$}
\DisplayProof
&
\AxiomC{$\mathcal{G} \vdash_E (w = \epsilon)$}
\RightLabel{\tiny $\top^* R$}
\UnaryInfC{$\myseq{\mathcal{G}}{\Gamma}{ w:\top^*;\Delta}$}
\DisplayProof\\[15px]
\multicolumn{2}{c}{
\AxiomC{$\myseq{(x,w\simp x');(y,y\simp w);(x,y\simp x');\Gcal}{\Gamma}{\Delta}$}
\RightLabel{\tiny $A_C$}
\UnaryInfC{$\myseq{(x,y\simp x');\Gcal}{\Gamma}{\Delta}$}
\DisplayProof
}\\[15px]
\multicolumn{2}{c}{
\AxiomC{$\myseq{(u,w \simp z);(y,v \simp w);(x,y \simp z);(u,v \simp x');\Gcal}{\Gamma}{\Delta}$}
\RightLabel{\tiny $A$}
\UnaryInfC{$\myseq{(x,y \simp z);(u,v \simp x');\Gcal}{\Gamma}{\Delta}$}
\DisplayProof
}\\[15px]
\multicolumn{2}{c}{
\AxiomC{$\myseq{(x,y\simp w');\mathcal{G}}{\Gamma}{x:A;w:A\mand B;\Delta}$}
\AxiomC{$\myseq{(x,y\simp w');\mathcal{G}}{\Gamma}{y:B;w:A\mand B;\Delta}$}
\RightLabel{\tiny $\mand R$}
\BinaryInfC{$\myseq{(x,y\simp w');\mathcal{G}}{\Gamma}{w:A\mand B;\Delta}$}
\DisplayProof
}\\[15px]
\multicolumn{2}{c}{
\AxiomC{$\myseq{(x,w'\simp z);\mathcal{G}}{\Gamma;w:A\mimp B}{x:A;\Delta}$}
\AxiomC{$\myseq{(x,w'\simp z);\mathcal{G}}{\Gamma;w:A\mimp B; z:B}{\Delta}$}
\RightLabel{\tiny $\mimp L$}
\BinaryInfC{$\myseq{(x,w'\simp z);\mathcal{G}}{\Gamma;w:A\mimp B}{\Delta}$}
\DisplayProof
}\\[15px]
\multicolumn{2}{l}{Side conditions:}\\[5px]
\end{tabular}
\begin{tabular}{l@{\extracolsep{1cm}}l}
\multicolumn{2}{l}{In $A,A_C$, the label $w$ does not occur in the conclusion.}\\
In $A_C$, $(x,y\simp x');\Gcal\vdash_E (x = x')$
&
In $A$, $(x,y \simp z);(u,v \simp x');\mathcal{G}\vdash_E (x = x')$\\
In $\mand R$, $(x,y\simp w');\mathcal{G} \vdash_E (w = w')$
&
In $\mimp L$, $(x,w'\simp z);\mathcal{G} \vdash_E (w = w')$
\end{tabular}
\caption{Changed rules in the intermediate system $\ilspslh$.}
\label{fig:ILSPSLH}
\end{figure*}

Given a set of relational atoms $\Gcal$, we define the relation
$=_\Gcal$ as follows: $a =_\Gcal b$ iff $\Gcal \vdash_E (a = b).$
We show next that $=_\Gcal$ is in fact an equivalence relation.
This equivalence relation will be useful in our counter-model construction later.

\begin{lemma}
\label{lm:eq_concat}
Let $\Gcal$ be a set of relational atoms, if $\Gcal\vdash_E (a = b)$
by applying $\sigma_1$ and $\Gcal\vdash_E(c = d)$ by applying
$\sigma_2$, then $\exists \sigma_3$
s.t. $\Scal(\Gcal,\sigma_1)\vdash_E (c\theta = d\theta)$ by
$\sigma_3$, where $\theta = subst(\sigma_1)$.
\end{lemma}
\begin{proof}
Note that $\Scal(\Gcal,\sigma_1) = \Gcal\theta$. So essentially we need to show that if $\Gcal \vdash_E (c = d)$, then $\Gcal\theta \vdash_E (c\theta = d\theta)$. This is a consequence of the substitution Lemma~\ref{lem:subs}.\qed
\end{proof}

\begin{lemma}
\label{lm:vdashe_eq}
Given a set of relational atoms $\Gcal$, the relation $=_\Gcal$ 
is an equivalence relation on the set of labels.
\end{lemma}
 
The intermediate system $\ilspslh$ is equivalent to $\lspslh$, i.e.,
every sequent provable in $\ilspslh$ is also provable
in $\lspslh$, and vice versa. 
This connection is easy to make, as is shown by H\'ou et al.~\cite{hou2013}. 
Properties such as contraction admissibility, closure
under substitution etc. also hold for $\ilspslh$.

\begin{lemma}
\label{lm:interm}
The intermediate 
labelled calculus $\ilspslh$ is equivalent to $\lspslh$.
\end{lemma}



%% file: counter_model_psl.tex
\subsection{Counter-model construction}
\label{subsec:counter_model_constr}

We now give a counter-model construction procedure for $\ilspslh$
which, by Lemma~\ref{lm:interm}, applies to $\lspslh$ as well. 
In the construction, we assume that labelled sequents
such as $\myseq{\Gcal}{\Gamma}{\Delta}$ are built from \textbf{sets}
$\Gcal,\Gamma,\Delta$ rather than \textbf{multisets}. This is harmless since
contraction is admissible in $\lspslh$ (and thus also in $\ilspslh$).
Detailed proofs in this section can be found in Appendix~\ref{sec:app2}.

As the counter-model construction involves infinite sets and sequents,
we extend the definition of $\vdash_E$ appropriately as below.

\begin{definition}
A (possibly infinite) set $\Gcal$ of relational atoms satisfies 
$\Gcal\vdash_E (x = y)$ iff $\Gcal_f\vdash_E (x = y)$ for some finite $\Gcal_f\subseteq \Gcal$.
\end{definition}

Given a set of relational atoms $\Gcal$, 
the equivalence relation $=_\Gcal$ partitions $\Lcal$ into
equivalence classes $[a]_\Gcal$ for each label $a \in \Lcal$:
$$[a]_\Gcal = \{ a' \in \Lcal \mid a =_\Gcal a' \}.$$


The counter-model procedure is essentially a procedure to saturate a sequent
by applying all applicable rules repeatedly. The 
aim
is to obtain 
an infinite saturated sequent from which a counter-model can be extracted. 
We first define a list of desired properties of such an infinite sequent which
would allow the counter-model construction. This is given in the following definition. 

\begin{definition}[Hintikka sequent]
\label{definition:hintikka_seq}
A labelled sequent $\myseq{\Gcal}{\Gamma}{\Delta}$ is a {\em Hintikka sequent}
if it satisfies the following
conditions for any formulae $A,B$ and any labels $a,a',b,c,d,e,z$:
\begin{enumerate} 
\item It is not the case that $a:A\in \Gamma$, $b:A\in\Delta$ and $a =_\Gcal b.$
\item If $a:A\land B\in \Gamma$ then $a:A\in \Gamma$ and $a:B\in \Gamma.$
\item If $a:A\land B\in \Delta$ then $a:A\in \Delta$ or $a:B\in \Delta.$
\item If $a:A\limp B\in \Gamma$ then $a:A\in \Delta$ or $a:B\in\Gamma.$
\item If $a:A\limp B\in \Delta$ then $a:A\in \Gamma$ and $a:B\in\Delta.$
\item If $a:\top^*\in \Gamma$ then $a =_\Gcal \epsilon.$
\item If $a:\top^* \in\Delta$ then $a \not =_\Gcal \epsilon.$

\item If $z:A\mand B\in \Gamma$ then $\exists x,y,z'$ s.t. 
$(x,y\simp z')\in\Gcal$, $z =_\Gcal z'$, $x:A\in\Gamma$ and $y:B\in\Gamma.$

\item If $z:A\mand B\in \Delta$ then $\forall x,y,z'$ 
if $(x,y\simp z')\in\Gcal$ and $z =_\Gcal z'$ then $x:A\in\Delta$ or $y:B\in\Delta.$

\item If $z:A\mimp B\in \Gamma$ then $\forall x,y,z'$ 
if $(x,z'\simp y)\in\Gcal$ and $z =_\Gcal z'$, then $x:A\in\Delta$ or $y:B\in\Gamma.$

\item If $z:A\mimp B\in \Delta$ then $\exists x,y,z'$ s.t. 
$(x,z'\simp y)\in\Gcal$, $z =_\Gcal z'$, $x:A\in\Gamma$ and $y:B\in\Delta.$

\item For any label $m\in \Lcal$, $(m,\epsilon\simp m)\in\Gcal.$

\item If $(a,b\simp c)\in\Gcal$ then $(b,a\simp c)\in\Gcal.$

\item If $(a,b\simp c)\in\Gcal$ and $(d,e\simp a')\in\Gcal$ and $a =_\Gcal a'$, 
then $\exists f,f'$ s.t. $(d,f\simp c)\in\Gcal$, $(b,e\simp
f')\in\Gcal$ and $f =_\Gcal f'$.
\item $a : \bot \not \in \Gamma$ and $a: \top \not \in \Delta.$
\end{enumerate}
\end{definition}

The next lemma 
shows that
a Hintikka
sequent gives a $\psl$ Kripke relational frame which is a (counter-)model of the formulae in the sequent.

\begin{lemma}
\label{lem:hintikka_sat}
Every Hintikka sequent is falsifiable. 
\end{lemma}


To prove the completeness of $\ilspslh$, we have to show that any given
unprovable sequent can be extended to a Hintikka sequent.
To do so we need a way to enumerate all possible applicable
rules in a fair way so that every rule will be chosen infinitely often. 
Traditionally, this is achieved via a fair enumeration strategy of every principal
formula of every rule. 
Since our calculus contains structural rules with no principal
formulas, we need to include them in the enumeration strategy as
well. For this purpose, we define a notion of {\em extended formulae}, given by the grammar:
\begin{center}
$ExF ::= F~|~\Ubb~|~\Ebb~|~\Abb~|~\Abb_C$
\end{center}
where $F$ is a formula, and $\Ubb,\Ebb,\Abb,\Abb_C$ are constants
that are used as ``dummy''
principal formulae for the structural rules $U$, $E$, $A$, and $A_C$, 
respectively.
A scheduler enumerates each combination of
left or right of turnstile, a label, an extended formula and at most two relational atoms infinitely often.

\newcommand{\exfml}[1]{Ex#1}

\begin{definition}[Scheduler $\phi$]
\label{definition:fair}
A {\em schedule} is a tuple $(O, m, \exfml{F}, R)$, where $O$ is either $0$
(left) or $1$ (right),
$m$ is a label, $\exfml{F}$ is an extended formula and $R$ is a set of relational atoms
such that $|R| \leq 2.$ 
Let $\Scal$ denote the set of all schedules.
A {\em scheduler} is a function from 
natural numbers $\Ncal$ to $\Scal.$ 
A scheduler $\phi$ is {\em fair} if 
for every schedule $S$, the set $\{i \mid \phi(i) = S \}$ is infinite.
\end{definition}

\begin{lemma}
There exists a fair scheduler.
\end{lemma}
\begin{proof}
Our proof is similar to the proof of \textit{fair strategy} of 
Larchey-Wendling~\cite{wendling2012}. To adapt their proof, 
we need to show
that the set $\Scal$ is countable. 
This follows from the fact that $\Scal$ is a finite product of countable sets. \qed
\end{proof}
From now on, we shall fix a fair scheduler, which we call 
$\phi$. 
We assume that the set of labels $\Lcal$
is totally ordered, and its elements can be enumerated
as $a_0,a_1,a_2,\ldots$ where $a_0 = \epsilon.$
This indexing 
is used
to select fresh labels in our construction of Hintikka sequents.

We say the formula $F$ is not cut-free provable in $\ilspslh$ if
the sequent $\vdash w : F$  is not cut-free derivable  in $\ilspslh$
for any label $w \not = \epsilon.$
Since we shall be concerned only with cut-free provability, in the following
when we mention derivation, we mean cut-free derivation. 

\begin{definition}
\label{definition:sequent-series}
Let $F$ be a formula which is not provable in $\ilspslh$. 
We construct a series of finite sequents 
$\langle \myseq{\Gcal_i}{\Gamma_i}{\Delta_i} \rangle_{i \in \Ncal}$
from $F$ where
$\Gcal_1 = \Gamma_1 = \emptyset$ and $\Delta_1 =
a_1:F$.

Assuming that $\myseq{\Gcal_i}{\Gamma_i}{\Delta_i}$
has been defined, we define $\myseq{\Gcal_{i+1}}{\Gamma_{i+1}}{\Delta_{i+1}}$ as follows.  
Suppose $\phi(i) = (O_i,m_i,\exfml{F}_i,R_i).$
\begin{itemize}

\item If $O_i = 0$, $\exfml{F}_i$ is a $\psl$ formula $C_i$ and
  $m_i:C_i\in\Gamma_i$:
\begin{itemize}
\item If $C_i = F_1\land F_2$, then $\Gcal_{i+1} = \Gcal_i$,
  $\Gamma_{i+1} = \Gamma_i\cup\{m_i:F_1,m_i:F_2\}$, $\Delta_{i+1} =
  \Delta_i$.
\item If $C_i = F_1\limp F_2$. If there is no derivation for
  $\myseq{\Gcal_i}{\Gamma_i}{m_i:F_1;\Delta_i}$ then $\Gamma_{i+1} =
  \Gamma_i$, $\Delta_{i+1} = \Delta_i\cup\{m_i:F_1\}$. Otherwise
  $\Gamma_{i+1} = \Gamma_i\cup\{m_i:F_2\}$, $\Delta_{i+1} =
  \Delta_i$. In both cases, $\Gcal_{i+1} = \Gcal_i$.
\item If $C_i = \top^*$, then $\Gcal_{i+1} =
  \Gcal_i\cup\{(\epsilon,m_i\simp \epsilon)\}$, $\Gamma_{i+1} =
  \Gamma_i$, $\Delta_{i+1} = \Delta_i$.
\item If $C_i = F_1\mand F_2$, then $\Gcal_{i+1} =
  \Gcal_i\cup\{(a_{2i},a_{2i+1}\simp m_i)\}$, $\Gamma_{i+1} =
  \Gamma_i\cup\{a_{2i}:F_1,a_{2i+1}:F_2\}$, $\Delta_{i+1} = \Delta_i$.
\item If $C_i = F_1\mimp F_2$ and $R_i = \{(x,m\simp  y)\}\subseteq\Gcal_i$ and $\Gcal_i\vdash_E (m = m_i)$. 
  If $\myseq{\Gcal_i}{\Gamma_i}{x:F_1;\Delta_i}$ has no derivation, then
  $\Gamma_{i+1} = \Gamma_i$, $\Delta_{i+1} =
  \Delta_i\cup\{x:F_1\}$. Otherwise $\Gamma_{i+1} =
  \Gamma_i\cup\{y:F_2\}$, $\Delta_{i+1} = \Delta_i$. In both cases,
  $\Gcal_{i+1} = \Gcal_i$.
\end{itemize}

\item If $O_i = 1$, $\exfml{F}_i$ is a $\psl$ formula $C_i$, and
  $m_i:C_i\in\Delta$:
\begin{itemize}
\item If $C_i = F_1\land F_2$. If there is no derivation for
  $\myseq{\Gcal_i}{\Gamma_i}{m_i:F_1;\Delta_i}$ then $\Delta_{i+1} =
  \Delta_i\cup\{m_i:F_1\}$. Otherwise $\Delta_{i+1} =
  \Delta_i\cup\{m_i:F_2\}$. In both cases, $\Gcal_{i+1} = \Gcal_i$ and
  $\Gamma_{i+1} = \Gamma_i$.
\item If $C_i = F_1\limp F_2$, then $\Gamma_{i+1} =
  \Gamma\cup\{m_i:F_1\}$, $\Delta_{i+1} = \Delta_i\cup\{m_i:F_2\}$,
  and $\Gcal_{i+1} = \Gcal_i$.
\item $C_i = F_1\mand F_2$ and $R_i = \{(x,y\simp
  m)\}\subseteq\Gcal_i$ and $\Gcal_i\vdash_E (m_i = m)$. If
  $\myseq{\Gcal_i}{\Gamma_i}{x:F_1;\Delta_i}$ has no derivation, then
  $\Delta_{i+1} = \Delta_i\cup\{x:F_1\}$. Otherwise $\Delta_{i+1} =
  \Delta_i\cup\{y:F_2\}$. In both cases, $\Gcal_{i+1} = \Gcal_i$ and
  $\Gamma_{i+1} = \Gamma_i$.
\item If $C_i = F_1\mimp F_2$, then $\Gcal_{i+1} =
  \Gcal_i\cup\{(a_{2i},m_i\simp a_{2i+1})\}$, $\Gamma_{i+1} =
  \Gamma_i\cup\{a_{2i}:F_1\}$, and $\Delta_{i+1} =
  \Delta_i\cup\{a_{2i+1}:F_2\}$.
\end{itemize}
\item If $\exfml{F}_i \in\{\Ubb,\Ebb,\Abb, \Abb_C\}$, we proceed as follows:
\begin{itemize}
\item If $\exfml{F}_i = \Ubb$, $R_i = \{(a_n,\epsilon\simp a_n)\}$, where $n
  \leq 2i+1$, then $\Gcal_{i+1}=\Gcal_i\cup\{(a_n,\epsilon\simp
  a_n)\}$, $\Gamma_{i+1} = \Gamma_i$, $\Delta_{i+1} = \Delta_i$.
\item If $\exfml{F}_i = \Ebb$, $R_i = \{(x,y\simp z)\}\subseteq\Gcal_i$, then
  $\Gcal_{i+1}=\Gcal_i\cup\{(y,x\simp z)\}$, $\Gamma_{i+1} =
  \Gamma_i$, $\Delta_{i+1} = \Delta_i$.
\item If $\exfml{F}_i = \Abb$, $R_i = \{(x,y\simp z);(u,v\simp
  x')\}\subseteq\Gcal_i$ and $\Gcal_i\vdash_E (x = x')$, then
  $\Gcal_{i+1}=\Gcal_i\cup\{(u,a_{2i}\simp z),(y,v\simp a_{2i})\}$,
  $\Gamma_{i+1} = \Gamma_i$, $\Delta_{i+1} = \Delta_i$.
\item If $\exfml{F}_i = \Abb_C$, $R_i = \{(x,y \simp x') \} \subseteq \Gcal_i$,
and $\Gcal_i \vdash_E (x = x')$ then $\Gcal_{i+1} = \Gcal_i \cup \{(x, a_{2i} \simp x), (y, y \simp a_{2i})\}$,
$\Gamma_{i+1} = \Gamma_i$, $\Delta_{i+1} = \Delta_i.$
\end{itemize}
\item In all other cases, $\Gcal_{i+1} = \Gcal_i$, $\Gamma_{i+1} =
  \Gamma_i$ and $\Delta_{i+1} = \Delta_i$.
\end{itemize}
\end{definition}

Intuitively, each tuple $(O_i, m_i, \exfml{F_i}, R_i)$ corresponds to
a potential rule application . If the components of the rule
application are in the current sequent, we apply the corresponding
rule to these components. 
The indexing of labels guarantees that the
choice of $a_{2i}$ and $a_{2i+1}$ are always fresh for the sequent
$\myseq{\Gcal_i}{\Gamma_i}{\Delta_i}$.
The construction in Def.~\ref{definition:sequent-series} non-trivially extends a similar construction of Hintikka CSS due to Larchey-Wendling~\cite{wendling2012}, in addition to which we have to consider the cases for structural rules.

We say $\myseq{\Gcal'}{\Gamma'}{\Delta'} \subseteq
\myseq{\Gcal}{\Gamma}{\Delta}$ iff $\Gcal'\subseteq\Gcal$,
$\Gamma'\subseteq\Gamma$ and $\Delta'\subseteq\Delta$. A labelled
sequent $\myseq{\Gcal}{\Gamma}{\Delta}$ is \textit{finite} if
$\Gcal,\Gamma,\Delta$ are finite sets. Define $\myseq{\Gcal'}{\Gamma'}{\Delta'}
\subseteq_f \myseq{\Gcal}{\Gamma}{\Delta}$ iff
$\myseq{\Gcal'}{\Gamma'}{\Delta'} \subseteq \myseq{\Gcal}{\Gamma}{\Delta}$ and
$\myseq{\Gcal'}{\Gamma'}{\Delta'}$ is finite. If
$\myseq{\Gcal}{\Gamma}{\Delta}$ is a finite sequent, it is
\textit{consistent} iff it does not have a derivation in $\ilspslh$. 
A (possibly infinite) sequent $\myseq{\Gcal}{\Gamma}{\Delta}$
is \textit{finitely-consistent} iff every
$\myseq{\Gcal'}{\Gamma'}{\Delta'}\subseteq_f \myseq{\Gcal}{\Gamma}{\Delta}$ is
consistent.


We write $\Lcal_i$ for the set of labels occurring in the sequent
$\myseq{\Gcal_i}{\Gamma_i}{\Delta_i}$. Thus $\Lcal_1 = \{a_1\}$. 
The following lemma states some properties of the construction of 
the sequents $\myseq{\Gcal_i}{\Gamma_i}{\Delta_i}$, e.g., the labels $a_{2i},a_{2i+1}$ are always fresh for $\myseq{\Gcal_i}{\Gamma_i}{\Delta_i}$. 
This can be proved by a simple induction on $i.$
\begin{lemma}
\label{lm:construction}
For any $i\in\mathcal{N}$, the following properties hold:
\begin{enumerate}
\item $\myseq{\Gcal_i}{\Gamma_i}{\Delta_i}$ has no derivation
\item $\Lcal_i\subseteq \{a_0, a_1,\cdots,a_{2i-1}\}$
\item $\myseq{\Gcal_i}{\Gamma_i}{\Delta_i}\subseteq_f
  \myseq{\Gcal_{i+1}}{\Gamma_{i+1}}{\Delta_{i+1}}$
\end{enumerate}
\end{lemma}


Given the construction of the series of sequents we have just seen above, we define
a notion of a {\em limit sequent}, as the union of 
every sequent in the series.

\begin{definition}[Limit sequent]
\label{definition:lim_seq} 
Let $F$ be a formula unprovable in $\ilspslh.$
The {\em limit sequent for $F$} is the sequent 
$\myseq{\Gcal^\omega}{\Gamma^\omega}{\Delta^\omega}$ 
where 
$\Gcal^\omega = \bigcup_{i\in\mathcal{N}}\Gcal_i$
and
$\Gamma^\omega = \bigcup_{i\in\mathcal{N}}\Gamma_i$
and
$\Delta^\omega = \bigcup_{i\in\mathcal{N}}\Delta_i$
and
where $\myseq{\Gcal_i}{\Gamma_i}{\Delta_i}$ is as defined in Def.\ref{definition:sequent-series}.
\end{definition}

The following lemma shows that the limit sequent defined above is
indeed a Hintikka sequent, thus we can use it to extract a
counter-model.

\begin{lemma}
\label{lem:lim_hintikka}
If $F$ is a formula unprovable in $\ilspslh$, then 
the limit labelled sequent for $F$ is a Hintikka sequent.
\end{lemma}


Finally we can state completeness: whenever a formula has
no derivation in $\ilspslh$,  we can extract an infinite
counter-model based on the limit sequent and the Kripke relational
frame.

\begin{theorem}[Completeness]
\label{thm:completeness}
Every formula $F$ unprovable in $\ilspslh$
is not valid (in $\psl$ relational Kripke models).
\end{theorem}
\begin{proof}
We construct a limit sequent $\myseq{\Gcal^\omega}{\Gamma^\omega}{\Delta^\omega}$ for $F$
following Def.~\ref{definition:lim_seq}. Note that by the construction of the limit sequent,
we have $a_1 : F \in \Delta^\omega.$ By Lemma~\ref{lem:lim_hintikka},
this limit sequent is a Hintikka sequent, and therefore by Lemma~\ref{lem:hintikka_sat},
$\myseq{\Gcal^\omega}{\Gamma^\omega}{\Delta^\omega}$ is falsifiable.
This means there exists a model $(\Fcal,\val,\rho)$ that satisfies $\Gcal^\omega$ and $\Gamma^\omega$
and falsifies every element of $\Delta^\omega$, including $a_1 : F$, which means that
$F$ is false at world $\rho(a_1).$  Thus $F$ is not valid.\qed
\end{proof}

\begin{corollary}
\label{cor:lspslh_complete}
If formula $F$ is unprovable in $\lspslh$
then $F$ is not valid (in $\psl$ relational Kripke models).
\end{corollary}

\begin{proof}
By Lemma~\ref{lm:interm} and Theorem~\ref{thm:completeness}.\qed
\end{proof}



%% file: extension_psl.tex
\section{Extensions of $\psl$}
\label{sec:extension_psl}
We now consider some extensions of $\psl$ obtained by imposing additional
properties on the semantics, as suggested by Dockins et al~\cite{dockins2009}.
We show that sound rules for \emph{indivisible unit} and the stronger property
of \emph{disjointness} can be added to our labelled sequent calculus without
jeopardising our completeness proof, but that the more exotic properties of
\emph{splittability} and \emph{cross-split} are not fully compatible with our
current framework, they require non-trivial changes to the proofs in previous sections.
See Appendix~\ref{sec:app3} for the proofs in this section.

\paragraph{Indivisible unit.} The unit $\epsilon$ in a commutative monoid 
$(H,\circ,\epsilon)$ is indivisible iff the following holds for 
any $h_1,h_2\in H$: $\mbox{ if } h_1\circ h_2 = \epsilon \mbox{ then } h_1 = \epsilon.$
Relationally, this corresponds to the first-order condition: $\forall h_1,h_2\in H. \mbox{ if } R(h_1, h_2, \epsilon) \mbox{ then } h_1 = \epsilon.$
Note that this also means that $h_2 = \epsilon$ whenever $h_1 \circ h_2 = \epsilon.$
Most memory models in the literature obey indivisible
unit~\cite{brotherston2010}, so this property seems appropriate
for reasoning about concrete applications of
separation logic. Indivisible unit can be axiomatised by the
formula~\cite{brotherston2013}:
$\top^*\land (A\mand B)\limp A$.
We use the following sound rule to capture this property:
\begin{center}
\AxiomC{$\myseq{(\epsilon,y \simp \epsilon);\Gcal[\epsilon/x]}{\Gamma[\epsilon/x]}{\Delta[\epsilon/x]}$}
\RightLabel{\tiny $IU$}
\UnaryInfC{$\myseq{(x,y\simp \epsilon);\Gcal}{\Gamma}{\Delta}$}
\DisplayProof
\end{center}

Note that we can then instantiate the label $y$ to $\epsilon$ by applying $Eq_1$
upwards.
Recall that the sequent calculus $\lsbbi$~\cite{hou2013} is 
just the sequent calculus $\lspslh$ minus the rules $C$ and $P$.
\begin{proposition}
\label{prop:iu_axiom}
The formula $\top^*\land (A\mand B)\limp A$ is provable in $\lsbbi + IU$. 
\end{proposition}


\begin{theorem}
\label{thm:add_iu}
$\lspslh + IU$ is sound and cut-free complete with respect to the
class of $\psl$ Kripke relational frames 
(and separation algebras)  
with indivisible unit.
\end{theorem}


\paragraph{Disjointness.} The separating conjunction $\mand$ in separation logic requires
that the two combined heaps have disjoint domains~\cite{reynolds2002}. 
In a separation algebra $(H, \circ, \epsilon)$, disjointness  
is defined by the following additional requirement: $\forall h_1,h_2\in H. \mbox{ if } h_1\circ h_1 = h_2 \mbox{ then } h_1 = \epsilon$.
Relationally:
$\forall h_1,h_2\in H. \mbox{ if } R(h_1, h_1, h_2) \mbox{ then } h_1 = \epsilon$.
This condition is captured 
by the following rule, where $x,y$ are labels.
\begin{center}
\AxiomC{$(\myseq{\epsilon,\epsilon\simp y);\Gcal[\epsilon/x]}{\Gamma[\epsilon/x]}{\Delta[\epsilon/x]}$}
\RightLabel{\tiny $D$}
\UnaryInfC{$\myseq{(x,x\simp y);\Gcal}{\Gamma}{\Delta}$}
\DisplayProof
\end{center}

Disjointness implies indivisible unit (but not vice versa), 
as shown by Dockins et al.~\cite{dockins2009}. We can prove the axiom for 
indivisible unit by using $\lsbbi + D$.

\begin{proposition}
\label{lem:d_iu_axiom}
The formula $\top^*\land (A\mand B)\limp A$ is provable in $\lsbbi +
D.$
\end{proposition}



\begin{theorem}
\label{thm:add_d}
$\lspslh + D$ is sound and cut-free complete with respect to the class
of $\psl$ Kripke relational frames 
(and separation algebras) 
with disjointness. 
\end{theorem}
\begin{proof}
Similar to Theorem~\ref{thm:add_iu}.\qed
\end{proof}

\paragraph{Splittability and cross-split.} The property of infinite splittability is sometimes
useful when reasoning about the kinds of resource sharing that occur in divide-and-conquer 
style computations~\cite{dockins2009}. A separation algebra $(H,\circ,\epsilon)$ has splittability if for 
every $h_0 \in H\setminus\{\epsilon\}$, there are $h_1, h_2 \in H\setminus \{\epsilon\}$ 
such that $h_1\circ h_2 = h_0$. 
Relationally,
if $h_0 \neq \epsilon$ then there exist $h_1 \neq \epsilon, h_2
\neq\epsilon$ s.t.\ $R(h_1, h_2, h_0)$.
This property can be axiomatised as the 
formula $\lnot\top^* \limp (\lnot\top^* \mand \lnot \top^*)$~\cite{brotherston2013}.
We give the following rules for splittability:
\begin{center}
\AxiomC{$\myseq{(x,y\simp z);(x\not =\epsilon);(y\not =\epsilon);(z\not =\epsilon);\Gcal}{\Gamma}{\Delta}$}
\RightLabel{\tiny $S$}
\UnaryInfC{$\myseq{(z \not = \epsilon);\Gcal}{\Gamma}{\Delta}$}
\DisplayProof
$\qquad$
\AxiomC{}
\RightLabel{\tiny $\not = L$}
\UnaryInfC{$\myseq{(w\not = w);\Gcal}{\Gamma}{\Delta}$}
\DisplayProof\\[10px]
\AxiomC{$\myseq{(w\not = \epsilon);\Gcal}{\Gamma}{\Delta}$}
\AxiomC{$\myseq{(\epsilon,w\simp\epsilon);\Gcal}{\Gamma}{\Delta}$}
\RightLabel{\tiny $EM$}
\BinaryInfC{$\myseq{\Gcal}{\Gamma}{\Delta}$}
\DisplayProof
\end{center}


We add a new type of structure, namely \emph{inequality}, to our calculus. The inequality structures are grouped with relational atoms in $\Gcal$. The rule $S$ directly encodes the semantics of splittability. We then need another rule $\not = L$ to conclude that $(w\not = w)$, for any label $w$, cannot be valid. Finally, the rule $EM$, named as the law of excluded middle for $(w= \epsilon)\lor (w\not = \epsilon)$, is essentially a cut on $w:\top^*$.

The advantage of formulating the above as structural rules is that 
these rules do not require extra arguments 
in the cut-elimination proof.


Cross-split is a rather complicated property. It specifies that if a heap can be 
split in two different ways, then there should be intersections of these splittings. 
Formally, in a separation algebra $(H,\circ,\epsilon)$, if $h_1\circ h_2 = h_0$ and $h_3\circ h_4 = h_0$, 
then there should be four elements $h_{13},h_{14},h_{23},h_{24}$, informally representing the intersections 
$h_1\cap h_3$, $h_1\cap h_4$, $h_2\cap h_3$ and $h_2\cap h_4$ respectively, 
such that $h_{13}\circ h_{14} = h_1$, $h_{23}\circ h_{24} = h_2$, $h_{13}\circ h_{23} = h_3$, 
and $h_{14}\circ h_{24} = h_4$. The corresponding condition on Kripke relational
frames is obvious. The following sound rule naturally 
captures cross-split, where $p,q,s,t,u,v,x,y,z$ are labels:
\begin{center}
\AxiomC{$\myseq{(p,q\simp x);(p,s\simp u);(s,t\simp
  y);(q,t\simp v);(x,y\simp z);(u,v\simp z);\Gcal}{\Gamma}{\Delta}$}
\RightLabel{\tiny $CS$}
\UnaryInfC{$\myseq{(x,y\simp z);(u,v\simp z);\Gcal}{\Gamma}{\Delta}$}
\noLine
\UnaryInfC{The labels $p,q,s,t$ do not occur in the conclusion}
\DisplayProof\\[5px]
\end{center}
However, to ensure contraction admissibility, we need the following
special case for this rule where the two principal relational atoms are the same.
\begin{center}
\AxiomC{$\myseq{(p,q\simp x);(p,s\simp x);(s,t\simp
  y);(q,t\simp y);(x,y\simp z);\Gcal}{\Gamma}{\Delta}$}
\RightLabel{\tiny $CS_C$}
\UnaryInfC{$\myseq{(x,y\simp z);\Gcal}{\Gamma}{\Delta}$}
\noLine
\UnaryInfC{The labels $p,q,s,t$ do not occur in the conclusion}
\DisplayProof
\end{center}


We note that Reynolds' heap model~\cite{reynolds2002} falsifies splittability,
as heaps are finite objects that only non-trivially split finitely often.
On the other hand cross-split is true in the heap model; however we are not
aware of any
formulae whose proof requires this property. 


\begin{proposition}
The rules $S,\not= L, EM, CS, CS_C$ are sound.
\end{proposition}

 
\begin{proposition}
\label{prop:add_s}
The axiom $\lnot\top^* \limp (\lnot\top^* \mand \lnot \top^*)$ for
splittability is provable in $\lsbbi + \{S,\not = L,EM\}$.
\end{proposition} 

It is easy to check that the rules $S, \not = L, EM, CS, CS_C$ do not break cut-elimination.
Since the completeness proofs for splittability and cross-split both
involve modifications of our previous counter-model construction
method, we present them together in Appendix~\ref{sec:app3_sub}. 

\begin{theorem}
\label{thm:comp_s_cs}
$\lspslh + \{S,\not=L,EM,CS,CS_C\}$ is sound and cut-free complete with respect to the class of $\psl$ Kripke relational frames 
(and separation algebras) 
with splittability and cross-split.
\end{theorem}

Following previous work~\cite{dockins2009,brotherston2013}, we refer to the set of additional properties: cancellativity, partial-determinism, indivisible unit, disjointness, splittability, and cross-split (we assume single unit) as \emph{separation theory}. And the set of rules $C,P,IU,D,S,\not= L, EM, CS,CS_C$ is called $\lsst$. By a proof that appropriately incorporates the treatments for these properties in this section, we obtain the following result:

\begin{theorem}
\label{thm:comp_bbi_st}
$\lsbbi + \lsst$ is sound and cut-free complete with respect to the class
of BBI Kripke relational frames with separation theory.
\end{theorem}

Subsystems of $\lsbbi+\lsst$ are also easy to obtain, we leave them
for the reader to verify. We give some examples of subsystems in the next section.

\section{Example subsystems of $\lsbbi+\lsst$}

We now consider various labelled calculi obtained
by extending $\lsbbi$ with one or more structural rules that correspond to
partial-determinism ($P$), cancellativity ($C$),
indivisible unit (IU), and disjointness ($D$). 
Most of the results in this section either directly follow from the proofs 
in previous sections, or are easy adaptations.  As those conditions 
for monoids are often given in a modular way, e.g., 
in~\cite{dockins2009,brotherston2013}, it is not
surprising that our structural rules can also be added modularly to $\lsbbi$, since they
just simulate those conditions directly and individually in the labelled sequent calculus.

\paragraph{Calculi without cancellativity.}
Some notions of separation logic omit cancellativity~\cite{JensenBirkedal2012},
so dropping the rule $C$ in $\lspslh$ gives an interesting system. 
The proofs in Sec.~\ref{sec:complete_lspslh} still work
if we just insist on a partial commutative monoid, and drop $C$ in $\vdash_E$. 

\begin{theorem}
\label{theorem:lsbbi_p}
The labelled sequent calculus $\lsbbi + P$ is sound and cut-free complete with respect to the partial commutative monoidal semantics for BBI.
\end{theorem}

As a result, it is easy to obtain the following sound and complete labelled calculi 
for the corresponding semantics: $\lsbbi + P + IU$ and  $\lsbbi + P + D$. 
The proofs are similar to that for Theorem~\ref{theorem:add_iu}.

\paragraph{Calculi  without partial-determinism.} 
Similar to above, dropping partial-determinism gives another sound and
complete labelled calculus $\lsbbi + C$,
although we are not aware of any concrete models in separation logic that employ
this framework.

\begin{theorem}
\label{theorem:lsbbi_C}
The labelled sequent calculus $\lsbbi + C$ is sound and cut-free complete with respect to the cancellative commutative monoidal semantics for BBI.
\end{theorem}

Again, using a similar argument as in Theorem.~\ref{theorem:add_iu}, 
we can obtain sound and complete labelled calculi $\lsbbi + C + IU$ and $\lsbbi + C + D$.

\paragraph{Calculi without partial-determinism and cancellativity.}
The labelled calculus $\lsbbi + IU$ is sound and complete by 
Prop.~\ref{prop:iu_axiom}, and cut-elimination holds. 

\begin{theorem}
The labelled sequent calculus $\lsbbi + IU$ is sound and cut-free complete 
with respect to the commutative monoidal semantics for BBI with indivisible unit.
\end{theorem}

To prove the completeness of the calculus $\lsbbi + D$, we need to go
through the counter-model construction proof, since disjointness is
not axiomatisable. It is easy to check that the proofs in
Section~\ref{sec:complete_lspslh} do not break when we define
$\vdash_E$ by using $Eq_1,Eq_2,D$ only, and the Hintikka sequent then
gives the BBI Kripke relational frame that obeys disjointness. The
other proofs remain the same.
\begin{theorem}
The labelled sequent calculus $\lsbbi + D$ is sound and cut-free
complete with respect to the commutative monoidal semantics for BBI
with disjointness.
\end{theorem}


To summarise, our approach offers a sound and cut-free calculus for the
extension of BBI with every combination of the properties $P,C,IU,D$. The case
where none of the properties hold, i.e. regular BBI, have already been
solved~\cite{hou2013,park2013}. Omitting the cases covered by the implication
of $IU$ by $D$, this provides us with the following eleven labelled calculi:
$$
\begin{array}{lll}
\lsbbi + IU & \lsbbi + C & \lsbbi + D\\
\lsbbi + P & \multicolumn{2}{l}{\lspslh (= \lsbbi + P + C)} \\
\lsbbi + P + IU & \lsbbi + C + IU & \lspslh + IU\\
\lsbbi + P + D & \lsbbi + C + D & \lspslh +D
\end{array}
$$

The subsystems containing splittability or cross-split can be obtained
by incorporating the treatments in the proof for
Thm.~\ref{thm:comp_s_cs}. As there are too many combinations, we do
not show them here.



%% file: examples_sl.tex
\section{Implementation and experiment}
\label{sec:experiments}

We discuss here an implementation of the proof system $\lspslh + D$.
It turns out that the $A_C$ rule is admissible in this system; in fact
it is admissible in the subsystem $\lsbbi + C$, as shown next. 
So we do not implement the $A_C$ rule. 
See Appendix~\ref{sec:app4} for detailed proofs in this section.

\begin{proposition}
\label{prop:admissible_a_c}
The $A_C$ rule is admissible in $\lsbbi + C.$
\end{proposition} 

On the other hand, we restrict the rule $U$ to create the identity
relational atom $(w,\epsilon\simp w)$ only if $w$ occurs in the
conclusion (denoted as $U'$ below). This does not reduce the power of $\lspslh$, as will be
shown next.

\begin{lemma}
\label{lem:u_rule}
If $\myseq{\Gcal}{\Gamma}{\Delta}$ is derivable in $\lspslh$, then it is derivable 
in $\lspslh - U + U'$.
\end{lemma}

Our implementation is based on the following strategy when applying rules:
\begin{enumerate}
\item Try to close the branch by rules $id,\bot L,\top^* R,\top^* R$.
\item If (1) not applicable, apply all possible $Eq_1,Eq_2,P,C,IU,D$
  rules to unify labels
\footnote{Although $IU$ is admissible, we keep it because it simplifies proof search.}.
\item If (1-2) not applicable, apply invertible rules $\land L$, $\land R$, $\limp L$, $\limp R$, $\mand L$, $\mimp R$, $\top^* L$ in all possible ways.
\item If (1-3) not applicable, try $\mand R$ or $\mimp L$ by choosing existing relational atoms.
\item If none of the existing relational atoms are applicable, or all combinations of $\mand R,\mimp L$ formulae and relational atoms are already applied in (4), apply structural rules on the set $\Gcal_0$ of relational atoms in the sequent as follows.
\begin{enumerate}
\item Use $E$ to generate all commutative variants of existing relational atoms in $\Gcal_0$, giving a set $\Gcal_1$.
\item Apply $A$ for each applicable pair in $\Gcal_1$, generating a set $\Gcal_2$.
\item Use $U'$ to generate all identity relational atoms for each label in $\Gcal_2$, giving the set $\Gcal_3$.
\end{enumerate}
\item If none of above is applicable, fail.
\end{enumerate}
Step (2) is terminating, because each substitution eliminates a label,
and we only have finitely many labels. Step (5) is not applicable when $\Gcal_3 =
\Gcal_0$. It is also clear that step (5)
is terminating. Note that we forbid applications of the rule $A$ to
the pair $\{(x,y\simp z),(u,v\simp x)\}$ of relational atoms when
$\{(u,w\simp z),(y,v\simp w)\}$, for some label $w$, (or any
commutative variants of this pair, e.g., $\{(w,u\simp z);(v,y\simp w)\}$
) is already in the sequent. This is because the created
relational atoms in such an $A$ application can be unified to existing
ones by using rules $P$,$C$.

We view $\Gamma,\Delta$ in a sequent $\myseq{\Gcal}{\Gamma}{\Delta}$
as lists,
and each time a logical rule is
applied, we place the subformulae in the front of the list. Thus our
proof search has a ``focusing flavour'', that always tries to
decompose the subformulae of a principal formula if possible. 
To guarantee completeness, each time we apply a $\mand R$ or $\mimp L$
rule, the principal formula is moved to the end of the list, so that
each principal formula for non-determinism rules $\mand R,\mimp L$ is
considered fairly, i.e., applied in turn.

\begin{table*}[t!]
\footnotesize
\centering
\begin{tabular}{|l|l|l|l|l|l|}
\hline
& Formula & BBeye  & $\fvlsbbi$ & Separata \\
   &     & (opt) & (heuristic) &  \\
\hline
(1) & $(a \mimp b) \land (\top \mand (\top^* \land a)) \limp b$ & 0.076 & 0.002 & 0.002\\
(2) & $(\top^* \mimp \lnot (\lnot a \mand \top^*)) \limp a$ & 0.080 & 0.004 & 0.002\\
(3) & $\lnot ((a \mimp \lnot (a \mand b)) \land ((\lnot a \mimp \lnot b) \land b))$ & 0.064 & 0.003 & 0.002\\
(4) & $\top^* \limp ((a \mimp (b \mimp c)) \mimp ((a \mand b) \mimp c))$ & 0.060 & 0.003 & 0.002\\
(5) & $\top^* \limp ((a \mand (b \mand c)) \mimp ((a \mand b) \mand c))$ & 0.071 & 0.002 & 0.004\\
(6) & $\top^* \limp ((a \mand ((b \mimp e) \mand c)) \mimp ((a \mand (b \mimp e)) \mand c))$ & 0.107 & 0.004 & 0.008\\
(7) & $\lnot ((a \mimp \lnot (\lnot (d \mimp \lnot (a \mand (c \mand b))) \mand a)) \land c \mand (d \land (a \mand b)))$ & 0.058 & 0.002 & 0.006\\
(8) & $\lnot ((c \mand (d \mand e)) \land B)$ where & 0.047 & 0.002 & 0.013\\
 & $ B := ((a \mimp \lnot (\lnot (b \mimp \lnot (d \mand (e \mand c)))
\mand a)) \mand (b \land (a \mand \top)))$ & & &\\
(9) & $\lnot ( C \mand (d \land (a \mand (b \mand e))))$ where & 94.230
& 0.003 & 0.053\\
 & $C := ((a \mimp \lnot (\lnot (d \mimp \lnot ((c \mand e) \mand (b
\mand a))) \mand a)) \land c)$ & & & \\
(10) & $(a \mand (b \mand (c \mand d))) \limp (d \mand (c \mand (b \mand a)))$ & 0.030 & 0.004 & 0.002\\
(11) & $(a \mand (b \mand (c \mand d))) \limp (d \mand (b \mand (c \mand a)))$ & 0.173 & 0.002 & 0.002\\
(12) & $(a \mand (b \mand (c \mand (d \mand e)))) \limp (e \mand (d \mand (a \mand (b \mand c))))$ & 1.810 & 0.003 & 0.002\\
(13) & $(a \mand (b \mand (c \mand (d \mand e)))) \limp (e \mand (b \mand (a \mand (c \mand d))))$ & 144.802 & 0.003 & 0.002\\
(14) & $\top^* \limp (a \mand ((b \mimp e) \mand (c \mand d)) \mimp ((a \mand d) \mand (c \mand (b \mimp e))))$ & 6.445 & 0.003 & 0.044\\
(15) & $\lnot(\top^* \land (a \land (b \mand \lnot(c \mimp (\top^* \limp a)))))$ & timeout(1000s) & 0.003 & 0.003\\
(16) & $((D \limp (E \mimp (D \mand E))) \limp (b \mimp ((D \limp (E \mimp ((D \mand a) \mand a))) \mand b)))$, where  & 0.039 & 0.005 & 8.772\\
 & $D := \top^* \limp a$ and $E :=  a\mand a$ & & & \\
(17) & $((\top^* \limp (a \mimp (((a \mand (a \mimp b)) \mand \lnot b) \mimp (a \mand (a \mand ((a \mimp b) \mand \lnot b)))))) \limp$ & timeout(1000s) & fail & 49.584\\
 & $((((\top^* \mand a) \mand (a \mand ((a \mimp b) \mand \lnot b))) \limp (((a \mand a) \mand (a \mimp b)) \mand \lnot b)) \mand \top^*))$ & & &\\ 
(18) & $(F\mand F)\limp F$, where $F := \lnot(\top \mimp \lnot \top^*)$ & invalid & invalid & 0.004\\
(19) & $(\top^* \land (a \mand b)) \limp a$ & invalid & invalid & 0.003\\
\hline
\end{tabular}
\caption{Experimental results from the prover Separata.}
\label{tab:experiment1}
\end{table*}

We incorporate a number of optimisations in the proof search. (1)
Back-jumping~\cite{baader2003} is used to collect the ``unsatisfiable
core'' along each branch. When one premise of a binary rule has a
derivation, we try to derive the other premise only when the
unsatisfiable core is not included in that premise. 
(2) A search
strategy discussed by Park et al~\cite{park2013} is also adopted. For
$\mand R$ and $\mimp L$ applications, we forbid the search to consider
applying the rule twice with the same pair of principal formula and
principal relational atom, since the effect is the same as
contraction, which is admissible. 
(3) Previous work on theorem proving for BBI has shown that associativity of $\mand$ is
a source of inefficiency in proof search~\cite{park2013,hou2013}. We borrow
the idea of the heuristic method presented in~\cite{hou2013} to
quickly solve certain associativity instances. When we detect
$z:A\mand B$ on the right hand side of a sequent, we try to search for
possible worlds (labels) for the subformulae of $A,B$ in the sequent,
and construct a binary tree using these labels. For example, if we can
find $x:A$ and $y:B$ in the sequent, we will take $x,y$ as the
children of $z$. When we can build such a binary tree of labels,
the corresponding relational atoms given by the binary tree will be
used (if they are in the sequent) as the prioritised ones when
decomposing $z:A\mand B$ and its subformulae. Of course, without a
free-variable system, our handling of this heuristic method is just a
special case of the original one, but this approach can
speed up the search in certain cases.

The experiments in this paper are conducted on a Dell Optiplex 790 desktop with Intel CORE i7
2600 @ 3.4 GHz CPU and 8GB memory, running Ubuntu 13.04. The theorem provers are written in Ocaml.

We test our prover Separata for $\lspslh + D$ on the formulae listed in Table~\ref{tab:experiment1}; the times displayed are in seconds. We compare the results with provers for BBI, BBeye~\cite{park2013} and the incomplete heuristic-based $\fvlsbbi$~\cite{hou2013}, when the formula is valid in BBI. We run BBeye in an iterative deepening way, and the time counted for BBeye is the total time it spends. Formulae (1-14) are used by Park et al. to test their prover BBeye for BBI~\cite{park2013}.
We can see that for formulae (1-14) the performance of Separata is comparable with the heuristic based prover for $\fvlsbbi$. Both provers are generally faster than BBeye. 
Formula (15) is one that BBeye had trouble with~\cite{hou2013}, but Separata handles it trivially.
However, there are cases where BBeye is faster than Separata. We found
the example formula (16) from a set of testings on randomly generated
BBI theorems. Formula (17) is a converse example where a randomly
generated BBI theorem causes BBeye to time out and $\fvlsbbi$ with
heuristics to terminate within the timeout but without finding a proof
due to its incompleteness. 
Formula (18) is valid only when the monoid is partial~\cite{wendling2010}, and formula (19) is the axiom of indivisible unit. Some interesting cases for disjointness will be shown later.
We do not investigate the details in the performances between these provers because they are for different logics. We leave further optimisations for Separata as future work.


%% file: sl.tex
\section{Future work}
\label{sec:futurework}

In this paper we have focused on \emph{propositional} inference, but the
assertion language of separation logic is generally taken to include first-order
logic, usually extended with arithmetic, or at least equality. More importantly,
this language is interpreted only with respect to some \emph{concrete}
semantics, the most well-known of which is the original heap model of
Reynolds~\cite{reynolds2002}. We refer readers to that paper for a more careful
description of this model; for the purposes of this section we will remark that
\emph{values} range across the integers, and \emph{addresses} across some
specified subset of the integers; that \emph{heaps}
are finite partial functions from addresses to values; and that
\emph{expressions} are built up from variables (evaluated with respect to some
store), values, and the usual arithmetic operations.

The advantage of this model is that it supports the interpretation of the
\emph{points-to} predicate $\mapsto$, which allows direct reference to the
contents of the heap: $E\mapsto E'$ is satisfied by a heap iff it is a singleton
map sending the address specified by the expression $E$ to the value specified
by the expression $E'$%
.

The question for future research is whether our labelled sequent calculus and
implementation could be extended to reason about such concrete predicates; this
section presents preliminary work in this direction. While the full power of pointer
arithmetic is an important subject for future work, for the purpose of
this work we set arithmetic aside and let expressions range across store
variables $e,e_1,e_2,\ldots$ only, as is done for example by Berdine et
al~\cite{berdine2005b}. The rules for quantifiers are straightforward, e.g.:
\begin{center}
\AxiomC{$\myseq{\Gcal}{\Gamma;h:F[e/x]}{\Delta}$}
\RightLabel{\tiny $\exists L$}
\UnaryInfC{$\myseq{\Gcal}{\Gamma;h:\exists x.F}{\Delta}$}
\DisplayProof
$\quad$
\AxiomC{$\myseq{\Gcal}{\Gamma}{h:F[e/x];\Delta}$}
\RightLabel{\tiny $\exists R$}
\UnaryInfC{$\myseq{\Gcal}{\Gamma}{h:\exists x.F;\Delta}$}
\DisplayProof
\end{center}
where $e$ does not appear free in the conclusion of $\exists L$.

Equality between variables simply requires that they are assigned by the store
to the same value, giving rise to the rules
\begin{center}
\AxiomC{$\myseq{\Gcal}{\Gamma[e_2/e_1]}{\Delta[e_2/e_1]}$}
\RightLabel{\tiny $= L$}
\UnaryInfC{$\myseq{\Gcal}{\Gamma;h:e_1 = e_2}{\Delta}$}
\DisplayProof
$\quad$
\AxiomC{}
\RightLabel{\tiny $= R$}
\UnaryInfC{$\myseq{\Gcal}{\Gamma}{h:e = e;\Delta}$}
\DisplayProof
\end{center}

Points-to poses a more complex problem as it involves direct interaction
with the contents of heaps; Fig.~\ref{fig:pointsto} presents putative
labelled sequent rules for this predicate. The semantics of $e_1
\mapsto e_2$ first require that the heap be a singleton, which is a spatial
property that can be captured by abstract semantics: a `singleton' world
is not equal to the identity world $\epsilon$, and cannot be split into two
non-$\epsilon$ worlds. This motivates rules $\mapsto L_1$ and $\mapsto L_2$. The
rules $\mapsto L_3$ and $\mapsto L_4$ address the content of heaps:
$\mapsto L_3$ says that two heaps 
with the same address (value of $e_1$)
must be the same heap, and $\mapsto L_4$ says that a singleton heap makes a
unique assignment.
\begin{figure*}[ht!]
\centering
\begin{tabular}{cc}
\multicolumn{2}{c}{
\raisebox{-2ex}{
\AxiomC{}
\RightLabel{\tiny $\mapsto L_1$}
\UnaryInfC{$\myseq{\Gcal}{\Gamma;\epsilon:e_1\mapsto e_2}{\Delta}$}
\DisplayProof}
$\qquad$
\AxiomC{$\myseq{(\epsilon,h_0\simp h_0);\Gcal[\epsilon/h_1][h_0/h_2]}{\Gamma[\epsilon/h_1][h_0/h_2];h_0:e_1\mapsto e_2}{\Delta[\epsilon/h_1][h_0/h_2]}$}
\alwaysNoLine
\UnaryInfC{$\myseq{(h_0,\epsilon\simp h_0);\Gcal[\epsilon/h_2][h_0/h_1]}{\Gamma[\epsilon/h_2][h_0/h_1];h_0:e_1\mapsto e_2}{\Delta[\epsilon/h_2][h_0/h_1]}$}
\alwaysSingleLine
\RightLabel{\tiny $\mapsto L_2$}
\UnaryInfC{$\myseq{(h_1,h_2\simp h_0);\Gcal}{\Gamma;h_0:e_1\mapsto e_2}{\Delta}$}
\DisplayProof
}\\[25px]
\AxiomC{$\myseq{\Gcal[h/h']}{\Gamma[h/h'];h:e_1\mapsto e_2;h:e_1\mapsto e_3}{\Delta[h/h']}$}
\RightLabel{\tiny $\mapsto L_3$}
\UnaryInfC{$\myseq{\Gcal}{\Gamma;h:e_1\mapsto e_2;h':e_1\mapsto e_3}{\Delta}$}
\DisplayProof
&
\AxiomC{$\myseq{\Gcal}{\Gamma[e_1/e_3][e_2/e_4];h:e_1\mapsto e_2}{\Delta[e_1/e_3][e_2/e_4]}$}
\RightLabel{\tiny $\mapsto L_4$}
\UnaryInfC{$\myseq{\Gcal}{\Gamma;h:e_1\mapsto e_2;h:e_3\mapsto e_4}{\Delta}$}
\DisplayProof
\end{tabular}
\caption{Some rules for the predicate $\mapsto$ in separation logic.}
\label{fig:pointsto}
\end{figure*}

\begin{table*}[t!]
\footnotesize
\centering
\begin{tabular}{|l|l|l|}
\hline
& Formula & Separata+ \\
\hline
(1) & $((e_1 \mapsto e_2) \mand (e_1 \mapsto e_2)) \limp \bot$ & 0.004\\
(2) & $(((e1 \mapsto e2) \mand (e3 \mapsto e4)) \land ((e1 \mapsto e2) \mand (e5 \mapsto e6))) \limp ((e3 \mapsto e6) * \top)$ & 0.002\\
(3) & $(\exists x3 x2 x1.(((x3 \mapsto x2) \mand (x1 \mapsto e)) \land (x2 = x1))) \limp (\exists x4 x5.((x4 \mapsto x5) \mand (x5 \mapsto e)))$ & 0.001\\
(4) & $\lnot((e1 \mapsto e2) \mimp \lnot(e3 \mapsto e4)) \limp ((e1 = e3) \land ((e2 = e4) \land \top^*))$ & 0.004\\
(5) & $\lnot(((e1 \mapsto p) \mand (e2 \mapsto q)) \mimp \lnot(e3 \mapsto r)) \limp \lnot(((e1 \mapsto p) \mimp \lnot(\lnot((e2 \mapsto q) \mimp \lnot(e3 \mapsto r)))))$ & 0.002\\
(6) & $\lnot((e1 \mapsto p) \mimp \lnot(e2 \mapsto q)) \limp \lnot((e1 \mapsto p) \mimp \lnot((e2 \mapsto q) \land ((e1 \mapsto p) \mand \top)))$ & 0.003\\
\hline
\end{tabular}
\caption{Experimental results from the prover Separata+.}
\label{tab:experiment2}
\end{table*}

Our implementation of the calculus defined by adding these rules to
$LS_{PASL}+D$ is not complete w.r.t. Reynolds' semantics: for example it is unable to prove the
formula below, which is based on a property for septraction due to Vafeiadis and Parkinson~\cite{vafeiadis2007}, and is valid in the heap model:
\begin{equation}
\label{eq_unprovable}
\top^* \limp \lnot((e_1\mapsto e_2) \mimp \lnot(e_1\mapsto e_2))
\end{equation}
This formula essentially asserts
that a heap satisfying
$e_1\mapsto e_2$ is possible to construct, but our prover does not
support explicit heap construction.  Nevertheless this incomplete
calculus does support strikingly elegant proofs of non-trivial
separation logic inferences, such as the DISJOINT axiom
of Parkinson~\cite[Cha.~5.3]{parkinson2005}:
\begin{center}
\AxiomC{}
\RightLabel{\tiny $\mapsto L1$}
\UnaryInfC{$\myseq{(\epsilon,\epsilon\simp a_0)}{\epsilon:(e_1 \mapsto e_2)}{a_0:\bot} $}
\RightLabel{\tiny $D$}
\UnaryInfC{$\myseq{(a_1,a_1\simp a_0)}{a_1:(e_1 \mapsto e_2)}{a_0:\bot}$}
\RightLabel{\tiny $\mapsto L3$}
\UnaryInfC{$\myseq{(a_1,a_2\simp a_0)}{a_1:(e_1 \mapsto e_2);a_2:(e_1 \mapsto e_2)}{a_0:\bot}$}
\RightLabel{\tiny $\mand L$}
\UnaryInfC{$\myseq{~}{a_0:(e_1 \mapsto e_2) \mand (e_1 \mapsto e_2)}{ a_0:\bot} $}
\RightLabel{\tiny $\limp R$}
\UnaryInfC{$\myseq{~}{~}{a_0:((e_1 \mapsto e_2) \mand   (e_1 \mapsto e_2)) \limp   \bot}$}
\DisplayProof
\end{center}

Our experimental prover Separata+, extending $LS_{PASL}+D$ with the rules of this
section, has proved a number of tested SL formulae very rapidly; see
Table~\ref{tab:experiment2} for some examples. Formulae (1-3) are taken from Galmiche and
M\'{e}ry~\cite{galmiche2010}; in particular, the first is the DISJOINT axiom
proved above. Formulae (4-6) are taken from Vafeiadis and Parkinson's study of
magic wand's De Morgan dual \emph{septraction},
$\lnot (A \mimp \lnot B)$~\cite{vafeiadis2007}. These results present
encouraging evidence that the work of this paper may form the basis of practical
theorem proving for the assertion language of separation logic.

Dealing with splittability and cross-split in our labelled calculi is one of our next goals. We are also interested in extending the techniques of this paper to concrete
semantics other than Reynolds' heap models, such as those surveyed by Calcagno
et al~\cite{calcagno2007} and Jensen and Birkedal~\cite{JensenBirkedal2012}.

%% file: conclusion.tex
\section{Related work}
\label{sec:conclusion}

There are many more automated tools, formalisations, and logical embeddings for
separation logic than can reasonably be surveyed within the scope of this
conference paper. Almost all are not directly comparable to this paper because
they deal with separation logic for some \emph{concrete} semantics.

One exception to this concrete approach is Holfoot~\cite{Tuerk2009}, a HOL
mechanisation with support for automated reasoning about the `shape' of SL
specifications  -- exactly those aspects captured by abstract separation logic.
However, unlike Separata, Holfoot does not support magic wand. This is a common
restriction when automating any notion of SL, because $\mimp$ is a source of
undecidability~\cite{brochenin2012}.
Conversely, the mechanisations and
embeddings that do incorporate magic wand tend to give little thought to (semi-)
decision procedures. An important exception to this is the tableaux of
Galmiche and M\'{e}ry~\cite{galmiche2010}, which are designed for the decidable
fragment of the assertion language of concrete separation logic with $\mimp$
identified by Calcagno et al~\cite{Calcagno01}, but may also be extendable to
the full assertion language. These methods have not been implemented, and given
the difficulty of the development we expect that practical implementation would
be non-trivial.
Another partial exception to the trend to omit $\mimp$ is
SmallfootRG~\cite{SmallfootRG}, which supports automation yet includes
\emph{septraction}~\cite{vafeiadis2007}, the De Morgan dual of $\mimp$. However
SmallfootRG does not support additive negation nor implication, and so $\mimp$
cannot be recovered; indeed in this setting septraction is mere `syntactic
sugar' that can be eliminated.

The denigration of magic wand is not without cost, as the connective, while
surely less useful than $\mand$, has found application. A non-exhaustive list
follows: generating weakest preconditions via backwards
reasoning~\cite{IshtiaqOHearn01}; specifying
iterators~\cite{parkinson2005,neelakantan2006,haack09}; reasoning about
parallelism~\cite{Dodds2011}; and various applications of septraction, such as
the specification of iterators and buffers~\cite{Cherini09}. For a particularly
deeply developed example, see the correctness proof for the Schorr-Waite Graph
Marking Algorithm of Yang~\cite{YangPhD}, which involves non-trivial inferences
involving $\mimp$ (Lems. 78 and 79). These examples provide ample motivation to
build proof calculi and tool support that include magic wand. Undecidability,
which in any case is pervasive in program proof, should not deter us from
seeking practically useful automation. 

Our work builds upon the labelled sequent calculi for BBI of H\'{o}u
et al~\cite{hou2013}. Their prover $FVLS_{BBI}$ implements a
free-variable calculus for BBI but is incomplete. Our extensions to
H\'{o}u et al involves two main advances: first, a counter-model
construction necessary to prove completeness; second, our prover deals
with labelled sequents directly and (given certain fairness
assumptions) is a complete semi-decision procedure for PASL and its
variants.
%
The link between BBI and SL is also emphasised as
motivation by Park et al~\cite{park2013}, whose BBI prover BBeye was used for
comparisons in Sec.~\ref{sec:experiments}.
This work was recently refined by Lee and Park~\cite{Lee2013}, in work
independent to our own, to a labelled sequent calculus for Reynolds' heap model.
Their calculus, like ours, cannot prove the formula \eqref{eq_unprovable}%
\footnote{Confirmed by private communications with authors.}
and so is not complete for these semantics.
Also related, but so far not
implemented, are the tableaux for partial-deterministic BBI of Larchey-Wendling
and Galmiche~\cite{wendling2009,wendling2012}, which, as mentioned in the
introduction to this paper, are claimed to be extendable with cancellativity to
attain PASL via a ``rather involved'' proof. In contrast, the relative ease
with which certain properties can be added or removed from labelled sequent
calculi is an important benefit of our approach;
this advantage comes from structural rules which directly capture the conditions
on Kripke relational frames, and handle the equality of worlds by explicit
global substitutions.

Finally we note that the counter-model construction of this paper was necessary
to prove completeness because many of the properties we are interested in are
not BBI-axiomatisable, as proved by Brotherston and
Villard~\cite{brotherston2013}; that paper goes on to give a sound and complete
Hilbert axiomatisation of these properties by extending BBI  with techniques from hybrid logic.
Sequent calculus and proof search in this setting is another promising
future direction.

%% file: appendix.tex
\section{Appendix}

This section provides the details of the proofs in this paper. When presenting derivations, we use double line to indicate that the premise and the conclusion are equivalent; and use dotted line to indicate the use of lemmas, theorems, etc..

\subsection{Proofs in Section~\ref{subsec:cut-elim}}
\label{sec:app1}

\noindent Proof of Lemma~\ref{lem:subs}.

\begin{proof}
By induction on $ht(\Pi)$, we do a case analysis on the last rule applied in the derivation. There are three sub-cases: (1) neither $x$ nor $y$ is the label of the principal formula, (2) $y$ is the label of the principal formula, and (3) $x$ is the label of the principal formula. Most of the rules can be proved as in $\lsbbi$, here we only illustrate the new rules $P$ and $C$. Apparently they all fall into the first sub-case, as they are structural rules and there is no principal formula for them. 

If the last rule in $\Pi$ is $P$, which generally runs as below,
\begin{center}
\AxiomC{$\myseq{\Gcal[c/d];(a,b\simp c)}{\Gamma[c/d]}{\Delta[c/d]}$}
\RightLabel{\tiny $P$}
\UnaryInfC{$\myseq{\Gcal;(a,b\simp c);(a,b\simp d)}{\Gamma}{\Delta}$}
\DisplayProof
\end{center}
we further distinguish three cases: (1) $x\not = d$ and $x\not = c$; (2) $x = d$; and (3) $x = c$.
\begin{enumerate}
\item If $x\not = d$ and $x\not = c$, we need to consider three sub-cases:
\begin{enumerate}
\item If $y \not = d$ and $y \not = c$, then the two substitutions $[y/x]$ and $[c/d]$ do not interfere with each other, thus we can use the induction hypothesis to substitute $[y/x]$ and reorder the substitutions to obtain the desired derivation.
\begin{center}
\AxiomC{$\Pi'$}
\alwaysNoLine
\UnaryInfC{$\myseq{\Gcal[c/d][y/x];(a,b\simp c)[y/x]}{\Gamma[c/d][y/x]}{\Delta[c/d][y/x]}$}
\doubleLine
\UnaryInfC{$\myseq{\Gcal[y/x][c/d];(a,b\simp c)[y/x]}{\Gamma[y/x][c/d]}{\Delta[y/x][c/d]}$}
\alwaysSingleLine
\RightLabel{\tiny $P$}
\UnaryInfC{$\myseq{\Gcal[y/x];(a,b\simp c)[y/x];(a,b\simp d)[y/x]}{\Gamma[y/x]}{\Delta[y/x]}$}
\DisplayProof
\end{center} 
\item If $y = d$ we first use the induction hypothesis, substituting $[c/x]$, then obtain the following derivation:
\begin{center}
\AxiomC{$\Pi'$}
\alwaysNoLine
\UnaryInfC{$\myseq{\Gcal[c/d][c/x];(a,b\simp c)[c/x]}{\Gamma[c/d][c/x]}{\Delta[c/d][c/x]}$}
\doubleLine
\UnaryInfC{$\myseq{\Gcal[c/x][c/d];(a,b\simp c)[c/x]}{\Gamma[c/x][c/d]}{\Delta[c/x][c/d]}$}
\doubleLine
\UnaryInfC{$\myseq{\Gcal[d/x][c/d];(a,b\simp c)[d/x]}{\Gamma[d/x][c/d]}{\Delta[d/x][c/d]}$}
\alwaysSingleLine
\RightLabel{\tiny $P$}
\UnaryInfC{$\myseq{\Gcal[d/x];(a,b\simp c)[d/x];(a,b\simp d)[d/x]}{\Gamma[d/x]}{\Delta[d/x]}$}
\DisplayProof
\end{center}
\item If $y = c$, the proof is similar to above, without the second last step.
\end{enumerate}
\item If $x = d$, we consider three sub-cases:
\begin{enumerate}
\item If $y\not = c$ and $y \not = \epsilon$, we use the induction hypothesis to substitute $[c/y]$, and obtain the following derivation:
\begin{center}
\AxiomC{$\Pi'$}
\alwaysNoLine
\UnaryInfC{$\myseq{\Gcal[c/d][c/y];(a,b\simp c)[c/y]}{\Gamma[c/d][c/y]}{\Delta[c/d][c/y]}$}
\doubleLine
\UnaryInfC{$\myseq{\Gcal[y/d][c/y];(a,b\simp c)[c/y]}{\Gamma[y/d][c/y]}{\Delta[y/d][c/y]}$}
\alwaysSingleLine
\RightLabel{\tiny $P$}
\UnaryInfC{$\myseq{\Gcal[y/d];(a,b\simp c);(a,b\simp y)}{\Gamma[y/d]}{\Delta[y/d]}$}
\DisplayProof
\end{center}
\item If $y\not=c$ but $y =\epsilon$, then we use induction hypothesis to substitute $[\epsilon/c]$, and obtain the following derivation:
\begin{center}
\AxiomC{$\Pi'$}
\alwaysNoLine
\UnaryInfC{$\myseq{\Gcal[c/d][\epsilon/c];(a,b\simp \epsilon)}{\Gamma[c/d][\epsilon/c]}{\Delta[c/d][\epsilon/c]}$}
\doubleLine
\UnaryInfC{$\myseq{\Gcal[\epsilon/d][\epsilon/c];(a,b\simp \epsilon)}{\Gamma[\epsilon/d][\epsilon/c]}{\Delta[\epsilon/d][\epsilon/c]}$}
\alwaysSingleLine
\RightLabel{\tiny $P$}
\UnaryInfC{$\myseq{\Gcal[\epsilon/d];(a,b\simp c);(a,b\simp \epsilon)}{\Gamma[\epsilon/d]}{\Delta[\epsilon/d]}$}
\DisplayProof
\end{center}
\item If $y = c$, then the case is reduced to admissibility of weakening on relational atoms.
\end{enumerate}

\item If $x = c$, the cases are similar to those for $x = d$.
\end{enumerate}

If the last rule in $\Pi$ is $C$, the proof is analogous to the proof for $P$. \qed
\end{proof} \ \\

\noindent Proof of Lemma~\ref{lem:invert}.

\begin{proof}
The rules $P,C$ themselves are trivially invertible, since the inverted versions can be proved by using Lemma~\ref{lem:subs}. The invertibility of all other rules except for $\top^* L$ can be proved similarly as in $\lsbbi$. Here we show the proof for $\top^* L$ in $\lspslh$. We do an induction on the height of the derivation. Base case is the same as the proof for $\lsbbi$. For the inductive case, we illustrate the cases where the last rule in the derivation is $P$ or $C$. Assume w.l.o.g. that the principal formula for the rule $\top^* L$ is $x:\top^*$. 
\begin{enumerate}
\item If the last rule is $P$, which runs as below.
\begin{center}
\AxiomC{$\myseq{\Gcal[c/d];(a,b\simp c)}{\Gamma[c/d];x:\top^*[c/d]}{\Delta[c/d]}$}
\RightLabel{\tiny $P$}
\UnaryInfC{$\myseq{\Gcal;(a,b\simp c);(a,b\simp d)}{\Gamma;x:\top^*}{\Delta}$}
\DisplayProof
\end{center}
we distinguish three sub-cases:
\begin{enumerate}
\item If $x \not = d$ and $x \not = c$, then the substitutions $[\epsilon/x]$ and $[c/d]$ are independent, thus we can use the induction hypothesis and applied the rule $\top^* L$ (meanwhile switch the order of substitutions) to obtain the desired derivation.
\item If $x = d$, the original derivation is as follows.
\begin{center}
\AxiomC{$\Pi$}
\alwaysNoLine
\UnaryInfC{$\myseq{\Gcal[c/d];(a,b\simp c)}{c:\top^*;\Gamma[c/d]}{\Delta[c/d]}$}
\alwaysSingleLine
\RightLabel{\tiny $P$}
\UnaryInfC{$\myseq{\Gcal;(a,b\simp c);(a,b\simp d)}{d:\top^*;\Gamma}{\Delta}$}
\DisplayProof
\end{center}
We apply the induction hypothesis on the premise, then apply $P$ to obtain the following derivation:
\begin{center}
\AxiomC{$\Pi'$}
\alwaysNoLine
\UnaryInfC{$\myseq{\Gcal[c/d][\epsilon/c];(a,b\simp \epsilon)}{\Gamma[c/d][\epsilon/c]}{\Delta[c/d][\epsilon/c]}$}
\doubleLine
\UnaryInfC{$\myseq{\Gcal[\epsilon/d][\epsilon/c];(a,b\simp \epsilon)}{\Gamma[\epsilon/d][\epsilon/c]}{\Delta[\epsilon/d][\epsilon/c]}$}
\alwaysSingleLine
\RightLabel{\tiny $P$}
\UnaryInfC{$\myseq{\Gcal[\epsilon/d];(a,b\simp c);(a,b\simp \epsilon)}{\Gamma[\epsilon/d]}{\Delta[\epsilon/d]}$}
\DisplayProof
\end{center}
\item If $x = c$, the case is similar.
\end{enumerate}
\item The case where the last rule is $C$ is similar to above.
\end{enumerate}
\qed
\end{proof}

\subsection{Proofs in Section~\ref{subsec:interm}}
\label{sec:app1b}

\noindent Proof of Lemma~\ref{lm:vdashe_eq}.
\begin{proof}
We show that $\vdash_E$ satisfies the following conditions:
\begin{description}
\item[Reflexivity:] for any label $a$ that occurs in $\Gcal$, we have 
$\Gcal\vdash_E (a = a)$ by applying an empty sequence of $Eq_1,Eq_2,P,C$ rules.

\item[Symmetry:] if $\Gcal\vdash_E (x = y)$, via a sequence $\sigma$
  of $Eq_1,Eq_2,P,C$ applications. Let $\theta = subst(\sigma)$, then
  by definition $x\theta\equiv y\theta$ in $\Gcal\theta$. Thus
  $y\theta \equiv x\theta$, and we obtain that $\Gcal\vdash_E (y= x)$.

\item[Transitivity:] if $\Gcal\vdash_E (x = y)$ and $\Gcal\vdash_E (y
  = z)$, then by Lemma~\ref{lm:eq_concat} we obtain a sequence
  $\sigma$ of $Eq_1,Eq_2,P,C$ applications, and let
  $\theta=subst(\sigma)$, then $x\theta \equiv y\theta \equiv
  z\theta$. Thus $\Gcal\vdash_E (x = z)$.
\end{description}
\vspace{-4ex}
\qed
\end{proof}

\subsection{Proofs in Section~\ref{subsec:counter_model_constr}}
\label{sec:app2}

\noindent Proof of Lemma~\ref{lem:hintikka_sat}.

\begin{proof}
Let $\myseq{\Gcal}{\Gamma}{\Delta}$ be a Hintikka sequent.
We construct an extended model
$\Mcal = (H, \simp_\Gcal, \epsilon_\Gcal, \val, \rho)$ as follows: 
\begin{itemize}
\item $H = \{[a]_\Gcal \mid a \in \Lcal \}$
\item $\simp_\Gcal([a]_\Gcal,[b]_\Gcal, [c]_\Gcal)$ 
iff $\exists a', b', c'.(a',b'\simp c') \in \Gcal$, 
$a =_\Gcal a'$, 
$b =_\Gcal b'$, 
$c =_\Gcal c'$
\item $\epsilon_\Gcal = [\epsilon]_\Gcal$
\item $\val(p) = \{ [a]_\Gcal \mid a : p \in \Gamma \}$ for every $p \in Var$
\item $\rho(a) = [a]_\Gcal$ for every $a \in \Lcal$
\end{itemize}

To reduce clutter, we shall drop the subscript $\Gcal$ 
in $[a]_\Gcal$ and 
write $[a], [b] \simp_\Gcal [c]$ instead
of $\simp_\Gcal([a],[b],[c]).$

We first show that $\Fcal = \myframe{H}{\simp_{\Gcal}}{\epsilon_\Gcal}$ is a $\psl$ Kripke relational frame.

\begin{description}
\item[identity:] 
for each $[a]\in H$, by definition, there must be a label $a'\in \Lcal$ such that $[a] = [a']$. 
It follows from condition 12 in Def.~\ref{definition:hintikka_seq} that $(a',\epsilon\simp a')\in\Gcal$, thus $[a],[\epsilon]\simp_{\Gcal} [a]$ holds.
\item[commutativity:] if $[a],[b]\simp_{\Gcal} [c]$ holds, there must
  be some $(a',b'\simp c')\in\Gcal$
  s.t. $[a]=[a'],[b]=[b'],[c]=[c']$. Then by condition 13 in
  Def.~\ref{definition:hintikka_seq}, $(b',a'\simp c')\in\Gcal$, therefore
  $[b],[a]\simp_{\Gcal} [c]$ holds.
\item[associativity:] if $[a],[b]\simp_{\Gcal} [c]$ and
  $[d],[e]\simp_{\Gcal} [a]$ holds, then there exist some $(a',b'\simp
  c')\in\Gcal$ and $(d',e'\simp a'')\in\Gcal$
  s.t. $[a]=[a']=[a''],[b]=[b'],[c]=[c'],[d]=[d'],[e]=[e']$. Then by
  condition 14 in Def.~\ref{definition:hintikka_seq}, there also exist labels
  $f,f'$ s.t. $(d',f\simp c')\in\Gcal$ and $(b',e'\simp f')\in\Gcal$
  and $[f]=[f']$. Thus we can find $[f]$ s.t. $[d],[f]\simp_{\Gcal}
  [c]$ and $[b],[e]\simp_{\Gcal} [f]$ hold.
\item[Partial-determinism:] If $[a],[b]\simp_{\Gcal} [c]$ and
  $[a],[b]\simp_{\Gcal} [d]$ hold, then there exists some $(a',b'\simp
  c')\in\Gcal$ and $(a'',b''\simp d')\in\Gcal$ s.t. $[a] =[a'] =[a''],
  [b]= [b']=[b''], [c]=[c'], [d]=[d']$. Then by
  Lemma~\ref{lm:eq_concat}, $\Gcal\vdash_E (c' = d')$ by using
  rule $P$ to unify $c'$ and $d'$, thus we obtain that $[c]=[c']=[d] =[d']$.

\item[Cancellativity:] if $[a],[b]\simp_{\Gcal} [c]$ and
  $[a],[d]\simp_{\Gcal} [c]$ hold, then we can find some $(a',b'\simp
  c')\in\Gcal$ and $(a'',d'\simp c'')\in\Gcal$ s.t. $[a] =[a'] =[a''],
  [c]= [c']=[c''], [b]=[b'], [d]=[d']$. Then by
  Lemma~\ref{lm:eq_concat}, $\Gcal\vdash_E (b' = c')$ by using $C$
  to unify $b'$ and $c'$, thus we obtain that $[b]=[b']=[c] =[c']$.
\end{description}
So $\Mcal$ is indeed a model based on a $\psl$ Kripke relational frame. 
We prove next that $\myseq{\Gcal}{\Gamma}{\Delta}$ is falsifiable
in $\Mcal.$
We need to show the following (where $\rho(m) = [m]$):
\begin{enumerate}
\item[(1)] If $(a, b \simp c) \in \Gcal$ then 
$([a], [b] \simp_\Gcal [c]).$
\item[(2)] If $m : A \in \Gamma$ then 
$\Mcal, \rho(m) \Vdash A.$
\item[(3)] If $m : A \in \Delta$ then
$\Mcal, \rho(m) \not \Vdash A.$
\end{enumerate}
Item (1) follows from the definition of $\simp_\Gcal$. 
We prove (2) and (3) simultaneously by induction on the size of $A$. 
In the following, to simplify presentation, we omit the $\Mcal$ from the 
forcing relation. 
\begin{description}
\item[Base cases:] when $A$ is an atomic proposition $p$.
\begin{itemize}
\item If $m:p\in\Gamma$ then $[m]\in \val(p)$ by definition of $\val$, so $[m] \Vdash p.$
\item Suppose $m:p\in\Delta$, but $[m] \Vdash p$. 
Then $m':p\in\Gamma$, for some $m'$ s.t. $m' =_\Gcal m$. 
This violates condition 1 in Def.~\ref{definition:hintikka_seq}. Thus $[m]\not\Vdash p$.
\end{itemize}
\item[Inductive cases:] when $A$ is a compound formula. We do a case
  analysis on the main connective of $A$.
\begin{itemize}
\item If $m:A\land B \in\Gamma$, by condition 2 in
  Def.~\ref{definition:hintikka_seq}, $m:A\in\Gamma$ and $m:B\in\Gamma$. By
  the induction hypothesis, $[m]\Vdash A$ and $[m]\Vdash B$, thus
  $[m]\Vdash A\land B$.
\item If $m:A\land B \in\Delta$, by condition 3 in
  Def.~\ref{definition:hintikka_seq}, $m:A\in\Delta$ or $m:B\in\Delta$. By
  the induction hypothesis, $[m]\not\Vdash A$ or $[m]\not\Vdash B$,
  thus $[m]\not\Vdash A\land B$
\item If $m:A\limp B\in\Gamma$, by condition 4 in
  Def.~\ref{definition:hintikka_seq}, $m:A\in\Delta$ or $m:B\in\Gamma$. By
  the induction hypothesis, $[m]\not\Vdash A$ or $[m]\Vdash B$, thus
  $[m]\Vdash A\limp B$.
\item If $m:A\limp B\in\Delta$, by condition 5 in
  Def.~\ref{definition:hintikka_seq}, $m:A\in\Gamma$ and $m:B\in\Delta$. By
  the induction hypothesis, $[m]\Vdash A$ and $[m]\not\Vdash B$, thus
  $[m]\not\Vdash A\limp B$.
\item If $m:\top^*\in\Gamma$ then  $[m] = [\epsilon]$ by condition 6 in
  Def.~\ref{definition:hintikka_seq}. Since
  $[\epsilon]\Vdash\top^*$, we obtain $[m]\Vdash\top^*$.
\item If $m:\top^*\in\Delta$, by condition 7 in
  Def.~\ref{definition:hintikka_seq}, $[m] \not = [\epsilon]$ and then
  $[m]\not\Vdash\top^*$.
\item If $m:A\mand B\in\Gamma$, by condition 8 in
  Def.~\ref{definition:hintikka_seq}, $\exists a,b,m'$ s.t. $(a,b\simp
  m')\in\Gcal$ and $[m] =[m']$ and $a:A\in\Gamma$ and
  $b:B\in\Gamma$. By the induction hypothesis, $[a]\Vdash A$ and
  $[b]\Vdash B$. Thus $[a],[b]\simp_{\Gcal} [m]$ holds and $[m]\Vdash
  A\mand B$.
\item If $m:A\mand B\in\Delta$, by condition 9 in
  Def.~\ref{definition:hintikka_seq}, $\forall a,b,m'$ if $(a,b\simp
  m')\in\Gcal$ and $[m]=[m']$, then $a:A\in\Delta$ or
  $b:B\in\Delta$. By the induction hypothesis, if such $a,b$ exist,
  then $[a]\not\Vdash A$ or $[b]\not\Vdash B$. For any
  $[a],[b]\simp_{\Gcal} [m]$, there must be some $(a',b'\simp
  m'')\in\Gcal$ s.t. $[a] = [a'], [b] = [b'], [m] = [m'']$. Then
  $[a]\not\Vdash A$ or $[b]\not\Vdash B$ therefore
  $[m]\not\Vdash A\mand B$.
\item If $m:A\mimp B\in\Gamma$, by condition 10 in
  Def.~\ref{definition:hintikka_seq}, $\forall a,b,m'$ if $(a,m'\simp
  b)\in\Gcal$ and $[m]=[m']$, then $a:A\in\Delta$ or
  $b:B\in\Gamma$. By the induction hypothesis, if such $a,b$ exists,
  then $[a]\not\Vdash A$ or $[b]\Vdash B$. Consider any
  $[a],[m]\simp_{\Gcal} [b]$, there must be some $(a',m''\simp
  b')\in\Gcal$ s.t. $[a] = [a']$, $[m''] = [m]$, and $[b] = [b']$. So $[a] \not\Vdash A$ or $[b]\Vdash B$,
  thus $[m]\Vdash A\mimp B$.
\item If $m:A\mimp B\in\Delta$, by condition 11 in
  Def.~\ref{definition:hintikka_seq}, $\exists a,b,m'$ s.t. $(a,m'\simp
  b)\in\Gcal$ and $[m]=[m']$ and $a:A\in\Gamma$ and $b:B\in\Delta$. By
  the induction hypothesis, $[a]\Vdash A$ and $[b]\not\Vdash B$ and
  $[a],[m]\simp_{\Gcal} [b]$ holds, thus $[m]\not\Vdash A\mimp B$.
\end{itemize}
\end{description}
\qed
\end{proof}\ \\

\noindent Proof of Lemma~\ref{lm:construction}.

\begin{proof}
Item 1 is based on the fact that the inference rules preserves
falsifiability upwards, and we always choose the branch with no
derivation.  To show item 2, we do an induction on $i$. Base case, $i = 1$, 
$\Lcal_1 \subseteq \{a_0, a_1\}$ (recall that $a_0 = \epsilon$).  Inductive cases: suppose item 2 holds for any
$i \leq n$, for $n+1$, we consider five cases depending on which rule
is applied on $\myseq{\Gcal_i}{\Gamma_i}{\Delta_i}$.
\begin{enumerate}
\item If $\mand L$ is applied, then 
$\Lcal_{i+1} = \Lcal_i\cup\{a_{2i},a_{2i+1}\}\subseteq\{a_1,\cdots,a_{2i+1}\}.$
\item If $\mimp R$ is applied, same as above.
\item If $U$ is applied, 
which generates $(a_n,\epsilon\simp a_n)$, then $n\leq 2i+1$, thus 
$\Lcal_{i+1} = \Lcal_i\cup\{a_n\}\subseteq\{a_1,\cdots,a_{2i+1}\}$.
\item If $A$ is applied, the fresh label in the premise is $a_{2i}$. 
Thus $\Lcal_{i+1} = \Lcal_i\cup\{a_{2i}\}\subseteq\{a_1,\cdots,a_{2i+1}\}$.
\item Otherwise, $\Lcal_{i+1} = \Lcal_i \subseteq \{a_1,\cdots,a_{2i+1}\}$.   
\end{enumerate}
Item 3 is obvious from the construction of $\myseq{\Gcal_{i+1}}{\Gamma_{i+1}}{\Delta_{i+1}}.$\qed
\end{proof} \ \\ 

\noindent Proof of Lemma~\ref{lem:lim_hintikka}.

\begin{proof}
Let $\myseq{\Gcal^\omega}{\Gamma^\omega}{\Delta^\omega}$ be the limit sequent. 
First we show that $\myseq{\Gcal^\omega}{\Gamma^\omega}{\Delta^\omega}$ is finitely-consistent. Consider any
$\myseq{\Gcal}{\Gamma}{\Delta}\subseteq_f \myseq{\Gcal^\omega}{\Gamma^\omega}{\Delta^\omega}$,
we show that $\myseq{\Gcal}{\Gamma}{\Delta}$ has no derivation. Since
$\Gcal,\Gamma,\Delta$ are finite sets, there exists $i\in
\mathcal{N}$ s.t. $\Gcal\subseteq \Gcal_i$, $\Gamma\subseteq
\Gamma_i$, and $\Delta\subseteq \Delta_i$. 
Moreover, $\myseq{\Gcal_i}{\Gamma_i}{\Delta_i}$ is not provable in $\ilspslh$.
Since weakening is admissible in $\ilspslh$,
$\myseq{\Gcal}{\Gamma}{\Delta}\subseteq_f
\myseq{\Gcal_i}{\Gamma_i}{\Delta_i}$ cannot be provable either.
So condition 1, 7, and 15 in Definition~\ref{definition:hintikka_seq} hold for the limit sequent, 
for otherwise we would be able to construct a provable finite labelled sequent from 
the limit sequent. We show the proofs that the other conditions 
in~Definition~\ref{definition:hintikka_seq} are also satisfied by the
limit sequent. The following cases are numbered according to items in Definition~\ref{definition:hintikka_seq}.
\begin{enumerate}
\setcounter{enumi}{1}
\item If $m:F_1\land F_2\in\Gamma^\omega$, then it is in some 
$\Gamma_i$, where $i\in\mathcal{N}$. Since $\phi$ select the formula infinitely often, 
there is $j > i$ such that $\phi(j) = (0,m,F_1\land F_2,R)$. Then by construction 
$\{m:F_1,m:F_2\}\subseteq \Gamma_{j+1}\subseteq \Gamma^\omega$. 

\item If $m:F_1\land F_2\in\Delta^\omega$, then it is in some $\Delta_i$, where $i\in\mathcal{N}$. 
Since $\phi$ select the formula infinitely often, there is $j > i$ such that
$\phi(j) = (1,m,F_1\land F_2,R)$. Then by construction $m:F_n\in\Delta_{j+1}\subseteq\Delta^\omega$, 
where $n\in\{1,2\}$ and $\myseq{\Gcal_j}{\Gamma_j}{m:F_n;\Delta_j}$ does not have a derivation.

\item If $m:F_1\limp F_2\in\Gamma^\omega$, similar to case 3.

\item If $m:F_1\limp F_2\in\Delta^\omega$, similar to case 2.

\item If $m:\top^*\in\Gamma^\omega$, then $m:\top^*\in\Gamma_i$, for some
  $i\in\mathcal{N}$, since each labelled formula from $\Gamma^\omega$ must appear somewhere
  in the sequence.
Then there exists $j > i$ such that $\phi(j) = (0,m,\top^*,R)$ where
this formula becomes principal.
By construction $(\epsilon,m\simp \epsilon)\in\Gcal_{j+1}\subseteq\Gcal^\omega$. 
Then $\Gcal^\omega\vdash_E (m = \epsilon)$ because $\Gcal_{j+1}\vdash_E (m = \epsilon)$. 
So $m =_{\Gcal^\omega} \epsilon$. 

\setcounter{enumi}{7}

\item\label{item:1} If $m:F_1\mand F_2\in\Gamma^\omega$, then it is in some $\Gamma_i$, where $i\in\mathcal{N}$. 
Then there exists $j > i$ such that $\phi(j) = (0,m,F_1\mand F_2,R)$. 
By construction $\Gcal_{j+1} = \Gcal_j\cup\{(a_{2j},a_{2j+1}\simp m)\}\subseteq\Gcal^\omega$, 
and $\Gamma_{j+1} = \Gamma_j\cup\{a_{2j}:F_1,a_{2j+1}:F_2\}\subseteq\Gamma^\omega$. 

\item\label{item:2} If $m:F_1\mand F_2\in\Delta^\omega$, then it is in some $\Delta_i$, where 
$i\in\mathcal{N}$. For any $(x,y\simp m')\in\Gcal^\omega$ such that $\Gcal^\omega\vdash_E(m = m')$, 
there exists $j > i$ such that $(x,y\simp m')\in\Gcal_j$ and $\Gcal_j\vdash_E (m = m')$. 
Also, there exists $k > j$ such that $\phi(k) = (1,m,F_1\mand
F_2,\{(x,y\simp m')\})$ where the labelled formula becomes principal.
Since $(x,y\simp m')\in\Gcal_k$ and $\Gcal_k\vdash (m = m')$,
we have either $x : F_1 \in \Delta_{k+1} \subseteq \Delta^\omega$ 
or $y : F_2 \in \Delta_{k+1} \subseteq \Delta^\omega.$

\item If $m:F_1\mimp F_2\in\Gamma^\omega$, similar to case~\ref{item:1}.

\item If $m:F_1\mimp F_2\in\Delta^\omega$, similar to case~\ref{item:2}.

\item For each $a_n\in \Lcal$, there is a $j \geq n$ such that 
$\phi(j) = (O,m,\Ubb,\{(a_n,\epsilon\simp a_n)\})$ where $U$ is
applied to $a_n$. Then 
$\Gcal_{j+1} = \Gcal_j\cup\{(a_n,\epsilon\simp a_n)\}\subseteq\Gcal^\omega$, because $n\leq 2j+1$. 

\item If $(x,y\simp z)\in\Gcal^\omega$, then it is in some $\Gcal_i$, where $i\in\mathcal{N}$. 
Then there is a $j > i$ such that $\phi(j) = (O,m,\Ebb,\{(x,y\simp
z)\})$ where $E$ is applied.
Then $\Gcal_{j+1} = \Gcal_j\cup\{(y,x\simp z)\}\subseteq\Gcal^\omega$.

\item If $(x,y\simp z) \in \Gcal^\omega$, $(u,v\simp x') \in \Gcal^\omega$, and $x =_{\Gcal^\omega} x'$, then there 
is some $\Gcal_i$, $i\in\mathcal{N}$ such that $\{(x,y\simp z),(u,v\simp x')\} \subseteq \Gcal_i$ 
and $\Gcal_i\vdash_E (x = x')$. 
There are two cases to consider, depending on whether $(x,y\simp z)$ and $(u,v\simp x')$ are the same
relational atoms. 
Suppose they are distinct. Then there must be some $j > i$ such that 
$\phi(j) = (O,m,\Abb,\{(x,y\simp z),(u,v\simp x')\})$. 
Then $\{(x,y\simp z),(u,v\simp x')\}\in\Gcal_j$ and $\Gcal_j\vdash_E (x = x')$. 
By construction we obtain that $\Gcal_{j+1} = \Gcal_j\cup\{(u,a_{2j}\simp z),(y,v\simp a_{2j})\}\subseteq\Gcal^\omega$. 
If $(x,y\simp z)$ and $(u,v\simp x')$ are the same
relational atom, then a similar argument can be applied, but in this case the
schedule to choose is one which selects $\Abb_C$ rather than $\Abb.$
\end{enumerate}
\qed
\end{proof}

\subsection{Proofs in Section~\ref{sec:extension_psl}}
\label{sec:app3}

\noindent Proof of Proposition~\ref{prop:iu_axiom}.
\begin{proof}
\ \\
\begin{center}
\AxiomC{}
\RightLabel{\tiny $id$}
\UnaryInfC{$\myseq{(\epsilon,a_2\simp \epsilon)}{\epsilon:A;a_2:B}{\epsilon:A}$}
\RightLabel{\tiny $IU$}
\UnaryInfC{$\myseq{(a_1,a_2\simp \epsilon)}{a_1:A;a_2:B}{\epsilon:A}$}
\RightLabel{\tiny $\top^* L$}
\UnaryInfC{$\myseq{(a_1,a_2\simp a_0)}{a_0:\top^*;a_1:A;a_2:B}{a_0:A}$}
\RightLabel{\tiny $\mand L$}
\UnaryInfC{$\myseq{}{a_0:\top^*;a_0:A\mand B}{a_0:A}$}
\RightLabel{\tiny $\land L$}
\UnaryInfC{$\myseq{}{a_0:\top^*\land (A\mand B)}{a_0:A}$}
\RightLabel{\tiny $\limp R$}
\UnaryInfC{$\myseq{~}{~}{a_0:(\top^*\land (A\mand B))\limp A}$}
\DisplayProof
\end{center}\qed
\end{proof}\ \\

\noindent Proof of Theorem~\ref{thm:add_iu}.

\begin{proof}
Soundness is straightforward as the rule $IU$ is essentially just
encodes the semantics into the labelled sequent calculus.

Cut-elimination follows by checking each lemmas in Section~\ref{subsec:cut-elim}. Specifically, we show the details of Lemma~\ref{lem:subs} (substitution) and Lemma~\ref{lem:invert} (invertibility) here. 
\begin{description}
\item[Substitution:]  Prove by induction on the height of the derivation, here we examine the case where $IU$ is the last rule in the derivation. The rule $IU$ is a structural rule, thus does not have a principal formula and belongs to the case where neither $x$ nor $y$ in the substitution $[y/x]$ is the label of the principal formula. We consider the subcases of $x,y$ being $a$ or not respectively, where the $IU$ application is shown below.
\begin{center}
\AxiomC{$\Pi$}
\noLine
\UnaryInfC{$\myseq{\Gcal[\epsilon/a];(\epsilon,b\simp \epsilon)}{\Gamma[\epsilon/a]}{\Delta[\epsilon/a]}$}
\RightLabel{\tiny $IU$}
\UnaryInfC{$\myseq{\Gcal;(a,b\simp \epsilon)}{\Gamma}{\Delta}$}
\DisplayProof
\end{center}
\begin{enumerate}
\item If $x \not =a$ we consider two sub-cases.
\begin{enumerate}
\item If $y \not = a$ then the substitutions $[y/x],[\epsilon/a]$ are independent, thus we can easily use the induction hypothesis to substitute $[y/x]$ and switch the order of substitutions to obtain the desired derivation.
If $x = b$ and $y = \epsilon$, the $IU$ application is reduced to $Eq_1$ and $E$ applications as follows, where $\Pi'$ is obtained by using the induction hypothesis to substitute $[\epsilon/b]$:
\begin{center}
\AxiomC{$\Pi'$}
\noLine
\UnaryInfC{$\myseq{\Gcal[\epsilon/a][\epsilon/b];(\epsilon,\epsilon\simp \epsilon)}{\Gamma[\epsilon/a][\epsilon/b]}{\Delta[\epsilon/a][\epsilon/b]}$}
\doubleLine
\UnaryInfC{$\myseq{\Gcal[\epsilon/b][\epsilon/a];(\epsilon,\epsilon\simp \epsilon)}{\Gamma[\epsilon/b][\epsilon/a]}{\Delta[\epsilon/b][\epsilon/a]}$}
\RightLabel{\tiny $Eq_1$}
\UnaryInfC{$\myseq{\Gcal[\epsilon/b];(a,\epsilon\simp \epsilon);(\epsilon,a\simp \epsilon)}{\Gamma[\epsilon/b]}{\Delta[\epsilon/b]}$}
\RightLabel{\tiny $E$}
\UnaryInfC{$\myseq{\Gcal[\epsilon/b];(a,\epsilon\simp \epsilon)}{\Gamma[\epsilon/b]}{\Delta[\epsilon/b]}$}
\DisplayProof
\end{center}

\item If $y = a$, we use the induction hypothesis to substitute $[\epsilon/x]$, then use $IU$ to obtain the derivation.
\begin{center}
\AxiomC{$\Pi'$}
\alwaysNoLine
\UnaryInfC{$\myseq{\Gcal[\epsilon/a][\epsilon/x];(\epsilon,b\simp \epsilon)}{\Gamma[\epsilon/a][\epsilon/x]}{\Delta[\epsilon/a][\epsilon/x]}$}
\doubleLine
\UnaryInfC{$\myseq{\Gcal[a/x][\epsilon/a];(\epsilon,b\simp \epsilon)}{\Gamma[a/x][\epsilon/a]}{\Delta[a/x][\epsilon/a]}$}
\alwaysSingleLine
\RightLabel{\tiny $IU$}
\UnaryInfC{$\myseq{\Gcal[a/x];(a,b\simp \epsilon)}{\Gamma[a/x]}{\Delta[a/x]}$}
\DisplayProof
\end{center} 
A special case where $x = b$ can be shown similarly.
\end{enumerate}
\item If $x = a$, again, we consider two sub-cases:
\begin{enumerate}
\item If $y\not = \epsilon$, we use the induction hypothesis to substitute $[\epsilon/y]$, and then use $IU$ to obtain the derivation.
\begin{center}
\AxiomC{$\Pi'$}
\alwaysNoLine
\UnaryInfC{$\myseq{\Gcal[\epsilon/a][\epsilon/y];(\epsilon,b\simp \epsilon)}{\Gamma[\epsilon/a][\epsilon/y]}{\Delta[\epsilon/a][\epsilon/y]}$}
\doubleLine
\UnaryInfC{$\myseq{\Gcal[y/a][\epsilon/y];(\epsilon,b\simp \epsilon)}{\Gamma[y/a][\epsilon/y]}{\Delta[y/a][\epsilon/y]}$}
\alwaysSingleLine
\RightLabel{\tiny $IU$}
\UnaryInfC{$\myseq{\Gcal[y/a];(y,b\simp \epsilon)}{\Gamma[y/a]}{\Delta[y/a]}$}
\DisplayProof
\end{center}

\item If $y = \epsilon$, then the substitution gives $\myseq{\Gcal[\epsilon/a];(\epsilon,b\simp \epsilon)}{\Gamma[\epsilon/a]}{\Delta[\epsilon/a]}$, which is known to be derivable by using $\Pi$.

\end{enumerate}

\end{enumerate}

\item[Invertibility:] The rule $IU$ is trivially invertible, as can be proved by using the substitution lemma. We show here (by induction on the height of the derivation) the case for the rule $\top^* L$, where the last rule in the derivation is $IU$. The other rules can be proved similarly as in~\cite{hou2013}. The last rule $IU$ runs as below.
\begin{center}
\AxiomC{$\Pi$}
\noLine
\UnaryInfC{$\myseq{\Gcal[\epsilon/a];(\epsilon,b\simp \epsilon)}{\Gamma[\epsilon/a];x:\top^*[\epsilon/a]}{\Delta[\epsilon/a]}$}
\RightLabel{\tiny $IU$}
\UnaryInfC{$\myseq{\Gcal;(a,b\simp \epsilon)}{\Gamma;x:\top^*}{\Delta}$}
\DisplayProof
\end{center}
we consider three sub-cases: (1) if $x\not = a$ and $x \not = b$, then we can safely apply the induction hypothesis on the premise, switch the order of substitutions, and apply $IU$ to obtain 
$\myseq{\Gcal[\epsilon/x];(a,b\simp \epsilon)}{\Gamma[\epsilon/x]}{\Delta[\epsilon/x]}$.
(2) If $x = a$, then the premise of the $IU$ application is what we need to derive (with a dummy $\epsilon:\top^*$ on the left hand side of the sequent). (3) If $x = b$, then $\myseq{\Gcal[\epsilon/b];(a,\epsilon\simp \epsilon)}{\Gamma[\epsilon/b]}{\Delta[\epsilon/b]}$ can be derived by applying $E$ and $Eq_1$ backwards then use the induction hypothesis to obtain the derivation. The details are the same as the derivation shown in 1(a) of the proof for substitution. 
\end{description}

Completeness can be proved via the same counter-model construction for
$\lspslh$ (Corollary~\ref{cor:lspslh_complete}). 
That is, we first define an intermediate calculus $\ilspslh$ $+ IU$ that is equivalent to
$\lspslh + IU$, and do counter-model construction in $\ilspslh + IU$. 
Since the $IU$ rule involves substitution, the rule will be
localised into the entailment relation $\vdash_E$, so the definition of $\vdash_E$
in Definition \ref{dfn:equiv-entail} is modified to include $IU$ in addition to 
$Eq_1,Eq_2,P$ and $C.$ Thus the rules of $\ilspslh + IU$ are exactly the same as
$\ilspslh$, and the only change is in the definition of $\vdash_E.$ The equivalence between
$\lspslh + IU$ and $\ilspslh + IU$ can be proved as in Lemma~\ref{lm:interm}.

Then we only need to show that a Hintikka sequent yields a Kripke relational frame 
that corresponds to a separation algebra with indivisible unit. 
In particular, no additional clauses are needed in the definition of
Hintikka sequent since it is parametric on the entailment relation $\vdash_E$.

For a Hintikka sequent $\myseq{\Gcal}{\Gamma}{\Delta}$, suppose
$(H,\simp_{\Gcal},[\epsilon])$ is the $\psl$ Kripke relational frame
generated by $\Gcal$. Given any $[a],[b]\simp_{\Gcal}[\epsilon]$, we
can find a $(a',b'\simp c')\in\Gcal$ such that $[a] =
[a'],[b]=[b'],[\epsilon]=[c']$. Also, we
can use the rule $IU$ to derive $\Gcal\vdash_E (a' =
\epsilon)$. Thus by Lemma~\ref{lm:eq_concat}, we obtain $[a]=[a']=[\epsilon]$. So the structure
$(H,\simp_{\Gcal},[\epsilon])$ generated by $\Gcal$ is indeed a $\psl$
Kripke relational frame that obeys indivisible unit.

The saturation with logical rules and structural rules
$\Ebb,\Ubb,\Abb, \Abb_C$ is then 
the same as in Section~\ref{sec:complete_lspslh}.\qed
\end{proof}\ \\

\noindent Proof of Proposition~\ref{lem:d_iu_axiom}.
\begin{proof}
we highlight the principal relational atoms where they are not obvious.
\begin{center}
\footnotesize
\AxiomC{}
\RightLabel{\tiny $id$}
\UnaryInfC{$\myseq{(\epsilon,\epsilon\simp \epsilon);\cdots}{\epsilon:A; \epsilon: B}{\epsilon: A}$}
\RightLabel{\tiny $Eq_1$}
\UnaryInfC{$\myseq{(\epsilon,a_1\simp \epsilon);\cdots}{a_1:A; \epsilon: B}{\epsilon: A}$}
\RightLabel{\tiny $E$}
%
\UnaryInfC{$\myseq{(a_1,w_2\simp w_1);$$(\epsilon,\epsilon\simp w_2);$\fbox{$(a_1,\epsilon\simp \epsilon)$};$\cdots}
{a_1:A; \epsilon: B}{\epsilon: A}$}
%
\RightLabel{\tiny $D$}
\UnaryInfC{$\myseq{(a_1,w_2\simp w_1);$\fbox{$(a_2,a_2\simp w_2)$}$;(a_1,a_2\simp \epsilon);\cdots}
{a_1:A; a_2: B}{ \epsilon: A}$}
\RightLabel{\tiny $A$}
\UnaryInfC{$\myseq{(a_1,w_1\simp \epsilon);$\fbox{$(\epsilon,a_2\simp w_1);(a_1,a_2\simp \epsilon)$}$;\cdots}
{a_1:A; a_2: B}{\epsilon: A}$}
\RightLabel{\tiny $A$}
\UnaryInfC{$\myseq{(\epsilon,\epsilon\simp \epsilon);(a_1,a_2\simp \epsilon)}{a_1:A; a_2: B}{\epsilon: A}$}
\RightLabel{\tiny $U$}
\UnaryInfC{$\myseq{(a_1,a_2\simp \epsilon)}{a_1:A; a_2: B}{\epsilon: A}$}
\RightLabel{\tiny $\top^* L$}
\UnaryInfC{$\myseq{(a_1,a_2\simp a_0)}{a_0: \top^*; a_1:A; a_2: B}{a_0: A}$}
\RightLabel{\tiny $\mand L$}
\UnaryInfC{$\myseq{~}{a_0: \top^*; a_0:A\mand B}{ a_0: A}$}
\RightLabel{\tiny $\land L$}
\UnaryInfC{$\myseq{~}{a_0: \top^*\land (A\mand B)}{a_0: A}$}
\RightLabel{\tiny $\limp R$}
\UnaryInfC{$\myseq{~}{~}{a_0: \top^*\land (A\mand B)\limp A}$}
\DisplayProof
\end{center}
\qed
\end{proof}\ \\

\noindent Proof of Proposition~\ref{prop:add_s}.

\begin{proof}
We start from $\lnot\top^* \limp (\lnot\top^* \mand \lnot \top^*)$, and 
obtain the following derivation backward:
\begin{center}
\small
\AxiomC{}
\RightLabel{\tiny $\top^* R$}
\UnaryInfC{$\myseq{(\epsilon,\epsilon\simp\epsilon)}{}{\epsilon:\top^*;\epsilon:\lnot\top^* \mand \lnot\top^*}$}
\RightLabel{\tiny $Eq_1$}
\UnaryInfC{$\myseq{(\epsilon,w\simp\epsilon)}{}{w:\top^*;w:\lnot\top^* \mand \lnot\top^*}$}
\AxiomC{$\Pi$}
\noLine
\UnaryInfC{$\myseq{(w\not = \epsilon)}{}{w:\top^*;w:\lnot\top^* \mand \lnot\top^*}$}
\RightLabel{\tiny $EM$}
\BinaryInfC{$\myseq{}{}{w:\top^*;w:\lnot\top^* \mand \lnot\top^*}$}
\RightLabel{\tiny $\lnot L$}
\UnaryInfC{$\myseq{}{w:\lnot \top^*}{w:\lnot\top^* \mand \lnot\top^*}$}
\RightLabel{\tiny $\limp R$}
\UnaryInfC{$\myseq{}{}{w:\lnot \top^* \limp (\lnot\top^* \mand \lnot\top^*)}$}
\DisplayProof
\end{center}
Where $\Pi$ is the following derivation:
\begin{center}
\small
\AxiomC{}
\RightLabel{\tiny $\not = L$}
\UnaryInfC{$\myseq{(\epsilon,\epsilon\simp\epsilon);(w\not = \epsilon);(\epsilon,y\simp w);(\epsilon\not = \epsilon);(y\not = \epsilon)}{}{w:\top^*;w:\lnot\top^* \mand \lnot\top^*}$}
\RightLabel{\tiny $Eq_1$}
\UnaryInfC{$\myseq{(\epsilon,x\simp\epsilon);(w\not = \epsilon);(x,y\simp w);(x\not = \epsilon);(y\not = \epsilon)}{}{w:\top^*;w:\lnot\top^* \mand \lnot\top^*}$}
\RightLabel{\tiny $\top^* L$}
\UnaryInfC{$\myseq{(w\not = \epsilon);(x,y\simp w);(x\not = \epsilon);(y\not = \epsilon)}{x:\top^*}{w:\top^*;w:\lnot\top^* \mand \lnot\top^*}$}
\RightLabel{\tiny $\lnot R$}
\UnaryInfC{$\myseq{(w\not = \epsilon);(x,y\simp w);(x\not = \epsilon);(y\not = \epsilon)}{}{x:\lnot\top^*;w:\top^*;w:\lnot\top^* \mand \lnot\top^*}$}
\AxiomC{$\Pi'$}
\RightLabel{\tiny $\mand R$}
\BinaryInfC{$\myseq{(w\not = \epsilon);(x,y\simp w);(x\not = \epsilon);(y\not = \epsilon)}{}{w:\top^*;w:\lnot\top^* \mand \lnot\top^*}$}
\RightLabel{\tiny $S$}
\UnaryInfC{$\myseq{(w\not = \epsilon)}{}{w:\top^*;w:\lnot\top^* \mand \lnot\top^*}$}
\DisplayProof
\end{center}
and $\Pi'$ is a symmetric derivation as the left branch of the $\mand R$ application.\qed
\end{proof}\ \\

\subsubsection{Proof of the completeness of $\lspslh$ plus splittability and cross-split.}
\label{sec:app3_sub}

This subsection proves Theorem~\ref{thm:comp_s_cs}.

The definition of the \emph{equivalence entailment} $\vdash_E$ is the same as Definition~\ref{dfn:equiv-entail}.
We then obtain the intermediate system $\ilspslhe$ as $\ilspslh$ plus $S$, $EM$, $CS_C$, and the following modifications:
\begin{center}
\AxiomC{$\Gcal\vdash_E (w = w')$}
\RightLabel{\tiny $\not = L$}
\UnaryInfC{$\myseq{(w\not = w');\Gcal}{\Gamma}{\Delta}$}
\DisplayProof\\[10px]
\AxiomC{$\myseq{(p,q\simp x);(p,s\simp u);(s,t\simp
  y);(q,t\simp v);(x,y\simp z);(u,v\simp z');\Gcal}{\Gamma}{\Delta}$}
\RightLabel{\tiny $CS$}
\UnaryInfC{$\myseq{(x,y\simp z);(u,v\simp z');\Gcal}{\Gamma}{\Delta}$}
\noLine
\UnaryInfC{$(x,y\simp z);(u,v\simp z');\Gcal\vdash_E (z = z')$}
\noLine
\UnaryInfC{The labels $p,q,s,t$ do not occur in the conclusion}
\DisplayProof
\end{center}
Note that $\vdash_E$ is a side condition instead of a premise.

The system $\ilspslhe$ is equivalent to $\lspslh + \{S,\not = L, EM, CS, CS_C\}$, which is straightforward to show. Thus in what follows we give a counter-model construction procedure for $\ilspslhe$, then obtain the completeness result for both systems. As $\ilspslhe$ is just an extension of $\ilspslh$, we only give the additional definitions and proofs, and the parts where modifications are made.

\begin{definition}[Hintikka sequent]
\label{definition:hitikka_seq2}
A labelled sequent $\myseq{\Gcal}{\Gamma}{\Delta}$ is a {\em Hintikka sequent}
if it satisfies the 
conditions in Definition~\ref{definition:hintikka_seq}
and the following,
for any formulae $A,B$ and any labels $a,a',b,c,d,e,z,z'$:
\begin{enumerate}
\setcounter{enumi}{15}
\item For any label $m\in\Lcal$, either $(m\not = \epsilon)\in\Gcal$ or $m =_\Gcal \epsilon$.
\item If $(z\not = \epsilon)\in\Gcal$, then $\exists x,y,$ s.t. $(x,y\simp z)\in\Gcal$, and $(x\not = \epsilon)\in\Gcal$, and $(y\not = \epsilon)\in\Gcal$.
\item It is not the case that $(a\not = a')\in\Gcal$ and $a =_\Gcal a'$.
\item If $(a,b\simp z)\in\Gcal$ and $(c,d\simp z')\in\Gcal$ and $z =_\Gcal z'$, then $\exists ac,bc,ad,bd$ s.t. $(ad,ac\simp a)\in\Gcal$, $(ad,bd\simp d)\in\Gcal$, $(ac,bc\simp c)\in\Gcal$, and $(bc,bd\simp b)\in\Gcal$.
\end{enumerate}
\end{definition}

We extend the proof for Lemma~\ref{lem:hintikka_sat} for the additional items in the above definition. That is, we show that every Hintikka sequent is satisfiable by additionally showing that the constructed model $(H,\simp_\Gcal,\epsilon_\Gcal,\nu,\rho)$ from the Hintikka sequent satisfies splittability and cross-split.

\begin{proof}
The following belongs to the first part of the proof for Lemma~\ref{lem:hintikka_sat}.
\begin{description}
\item[Splittability:] for each $[a]\in H$, there is some $a'\in\Lcal$ s.t. $[a]=[a']$. By condition 17 in Definition~\ref{definition:hitikka_seq2}, either (1) $(a'\not = \epsilon)\in\Gcal$ or (2) $a'=_\Gcal \epsilon$. If (1) holds, by condition 18 in Definition~\ref{definition:hitikka_seq2}, there exist $x,y$ s.t. $(x,y\simp a')\in\Gcal$, $(x\not = \epsilon)\in\Gcal$, and $(y\not = \epsilon)\in\Gcal$ hold. Thus we can find $[x],[y]$ s.t. $[x],[y]\simp_\Gcal [a]$ holds and $[x]\not = [\epsilon]$, $[y]\not = [\epsilon]$. If (2) holds, splittability trivially holds.
\item[Cross-split:] if $[a],[b]\simp_\Gcal [z]$ and $[c],[d]\simp_\Gcal [z]$ hold, then we can fine some $(a',b'\simp z')\in\Gcal$ and $(c',d'\simp z'')\in\Gcal$ s.t. $[a] = [a']$, $[b] = [b']$, $[c] = [c']$, $[d] = [d']$, and $[z] = [z'] = [z'']$. This implies that $z' =_\Gcal z''$. Then by condition 20 in Definition~\ref{definition:hitikka_seq2}, there are $ad$, $ac$, $bc$, $bd$, s.t. $(ad,ac\simp a')\in\Gcal$, $(ad,bd\simp d')\in\Gcal$, $(ac,bc\simp c')\in\Gcal$, and $(bc,bd\simp b')\in\Gcal$. Therefore we can find $[ad]$, $[ac]$, $[bc]$, $[bd]$, s.t. $[ad],[ac]\simp_\Gcal [a]$, $[ad],[bd]\simp_\Gcal [d]$, $[ac],[bc]\simp_\Gcal [c]$, and $[bc],[bd]\simp_\Gcal [b]$.
\end{description}\qed
\end{proof}

To deal with the new rules we added, we now define the \emph{extended formulae} as
\begin{center}
$ExF ::= F~|~\Ubb~|~\Ebb~|~\Abb~|~\Abb_C~|~\Sbb~|~\EMbb~|~\CSbb~|~\CSbb_C$
\end{center}
where $F$ is a BBI formula, and the others are constants. We also need to redefine the schedule to handle the inequality structure. Thus Definition~\ref{definition:fair} is now modified as below.

\begin{definition}[Scheduler $\phi$]
\label{definition:fair2}
A {\em schedule} is a tuple $(O, m, \exfml{F}, R, I)$, where $O$ is either $0$
(left) or $1$ (right),
$m$ is a label, $\exfml{F}$ is an extended formula, $R$ is a set of relational atoms
such that $|R| \leq 2$, and $I$ is a singleton inequality. 
Let $\Scal$ denote the set of all schedules.
A {\em scheduler} is a function from the set of natural numbers $\Ncal$ to $\Scal.$ 
A scheduler $\phi$ is {\em fair} if 
for every schedule $S\in\Scal$, the set $\{i \mid \phi(i) = S \}$ is infinite.
\end{definition}

It follows from the same reason that there exists a fair scheduler. The new component $I$ in the scheduler is ignored in all the cases in Definition~\ref{definition:sequent-series}. However, we have to slightly extend the definition to accommodate splittability and cross-split. The original definition is rewritten as follows.

\begin{definition}
\label{definition:sequent-series2}
Let $F$ be a formula which is not provable in $\ilspslhe$. 
We construct a series of finite sequents 
$\{\myseq{\Gcal_i}{\Gamma_i}{\Delta_i} \}_{i \in \Ncal}$
from $F$ where
$\Gcal_1 = \Gamma_1 = \emptyset$ and $\Delta_1 =
a_1:F$.

Assuming that $\myseq{\Gcal_i}{\Gamma_i}{\Delta_i}$
has been defined, we define $\myseq{\Gcal_{i+1}}{\Gamma_{i+1}}{\Delta_{i+1}}$ as follows.  
Suppose $\phi(i) = (O_i,m_i,\exfml{F}_i,R_i, I_i).$
\begin{itemize}

\item If $O_i = 0$, $\exfml{F}_i$ is a $\psl$ formula $C_i$ and
  $m_i:C_i\in\Gamma_i$:
\begin{itemize}
\item If $C_i = F_1\land F_2$, same as original def..
\item If $C_i = F_1\limp F_2$, same as original def.. 
\item If $C_i = \top^*$, same as original def..
\item If $C_i = F_1\mand F_2$, then $\Gcal_{i+1} =
  \Gcal_i\cup\{(a_{4i},a_{4i+1}\simp m_i)\}$, $\Gamma_{i+1} =
  \Gamma_i\cup\{a_{4i}:F_1,a_{4i+1}:F_2\}$, $\Delta_{i+1} = \Delta_i$.
\item If $C_i = F_1\mimp F_2$ and $R_i = \{(x,m\simp  y)\}\subseteq\Gcal_i$ and $\Gcal_i\vdash_E (m = m_i)$, same as original def.. 
\end{itemize}

\item If $O_i = 1$, $\exfml{F}_i$ is a $\psl$ formula $C_i$, and
  $m_i:C_i\in\Delta$:
\begin{itemize}
\item If $C_i = F_1\land F_2$, same as original def..
\item If $C_i = F_1\limp F_2$, same as original def..
\item $C_i = F_1\mand F_2$ and $R_i = \{(x,y\simp
  m)\}\subseteq\Gcal_i$ and $\Gcal_i\vdash_E (m_i = m)$, same as original def..
\item If $C_i = F_1\mimp F_2$, then $\Gcal_{i+1} =
  \Gcal_i\cup\{(a_{4i},m_i\simp a_{4i+1})\}$, $\Gamma_{i+1} =
  \Gamma_i\cup\{a_{4i}:F_1\}$, and $\Delta_{i+1} =
  \Delta_i\cup\{a_{4i+1}:F_2\}$.
\end{itemize}
\item If $\exfml{F}_i \in\{\Ubb,\Ebb,\Abb, \Abb_C, \Sbb,\EMbb,\CSbb,\CSbb_C\}$, we proceed as follows:
\begin{itemize}
\item If $\exfml{F}_i = \Ubb$, $R_i = \{(a_n,\epsilon\simp a_n)\}$, where $n
  \leq 4i+3$, then $\Gcal_{i+1}=\Gcal_i\cup\{(a_n,\epsilon\simp
  a_n)\}$, $\Gamma_{i+1} = \Gamma_i$, $\Delta_{i+1} = \Delta_i$.
\item If $\exfml{F}_i = \Ebb$, same as original def..
\item If $\exfml{F}_i = \Abb$, $R_i = \{(x,y\simp z);(u,v\simp
  x')\}\subseteq\Gcal_i$ and $\Gcal_i\vdash_E (x = x')$, then
  $\Gcal_{i+1}=\Gcal_i\cup\{(u,a_{4i}\simp z),(y,v\simp a_{4i})\}$,
  $\Gamma_{i+1} = \Gamma_i$, $\Delta_{i+1} = \Delta_i$.
\item If $\exfml{F}_i = \Abb_C$, $R_i = \{(x,y \simp x') \} \subseteq \Gcal_i$,
and $\Gcal_i \vdash_E (x = x')$ then $\Gcal_{i+1} = \Gcal_i \cup \{(x, a_{4i} \simp x), (y, y \simp a_{4i})\}$,
$\Gamma_{i+1} = \Gamma_i$, $\Delta_{i+1} = \Delta_i.$
\item If $\exfml{F}_i = \Sbb$ and $I_i = \{(w\not = \epsilon)\}\subseteq \Gcal_i$, then $\Gcal_{i+1}=\Gcal_i\cup\{(a_{4i},a_{4i+1}\simp w),(a_{4i}\not = \epsilon),(a_{4i+1}\not = \epsilon)\}$, $\Gamma_{i+1} = \Gamma_i$, $\Delta_{i+1} = \Delta_i$.
\item If $\exfml{F}_i = \EMbb$ and $R_i = \{(\epsilon,a_n\simp \epsilon)\}$, where $n\leq 4i+3$. If there is no derivation for $\myseq{(\epsilon,a_n\simp \epsilon);\Gcal_i}{\Gamma_i}{\Delta_i}$, then $\Gcal_{i+1} = \Gcal_i\cup\{(\epsilon,a_n\simp \epsilon)\}$, otherwise $\Gcal_{i+1} = \Gcal_i\cup\{(a_n\not = \epsilon)\}$. In both cases, $\Gamma_{i+1}=\Gcal_i$ and $\Delta_{i+1} = \Delta_i$.
\item If $\exfml{F}_i = \CSbb$, $R_i = \{(x,y\simp z),(u,v\simp z')\}\subseteq\Gcal_i$, and $\Gcal_i\vdash_E (z = z')$, then $\Gcal_{i+1} = \Gcal_i\cup\{(a_{4i},a_{4i+1}\simp x),(a_{4i},a_{4i+2}\simp u),(a_{4i+2},a_{4i+3}\simp y),(a_{4i+1},a_{4i+3}\simp v)\}$, $\Gamma_{i+1} = \Gamma_i$, $\Delta_{i+1} = \Delta_i$.
\item If $\exfml{F}_i = \CSbb_C$ and $R_i = \{(x,y\simp z)\}\subseteq\Gcal_i$, then $\Gcal_{i+1} = \Gcal_i\cup\{(a_{4i},a_{4i+1}\simp x),(a_{4i},a_{4i+2}\simp x),(a_{4i+2},a_{4i+3}\simp y),(a_{4i+1},a_{4i+3}\simp y)\}$, $\Gamma_{i+1} = \Gamma_i$, $\Delta_{i+1} = \Delta_i$. 
\end{itemize}
\item In all other cases, $\Gcal_{i+1} = \Gcal_i$, $\Gamma_{i+1} =
  \Gamma_i$ and $\Delta_{i+1} = \Delta_i$.
\end{itemize}
\end{definition}

Lemma~\ref{lm:construction} is easy to show for the new definitions, the second item in the lemma should now be stated as $\Lcal\subseteq\{a_0,a_1,\cdots,a_{4i-1}\}$. We are ready to prove Lemma~\ref{lem:lim_hintikka} for the new conditions in the Hintikka sequent.

\begin{proof}
Condition 18 holds because the limit sequent is finitely-consistent. We show the cases for conditions 16, 17, 19 as follows.
\begin{enumerate}
\setcounter{enumi}{15}
\item For each $a_n\in\Lcal$, there is some natural number $j \geq n$ s.t. $\phi(j) = (O,m, \EMbb,$ $\{(\epsilon,a_n\simp \epsilon)\},I)$, where $EM$ is applied to $a_n$. Then either (1) $\Gcal_{j+1} = \Gcal_j\cup\{(\epsilon,a_n\simp \epsilon)\}$ or (2) $\Gcal_{j+1} = \Gcal_j\cup\{(a_n\not = \epsilon)\}$, depending on which choice gives a finitely-consistent sequent $\myseq{\Gcal_{j+1}}{\Gamma_{j+1}}{\Delta_{j+1}}$. If (1) holds, then $(\epsilon,a_n\simp \epsilon)\in\Gcal_{j+1}\subseteq \Gcal^\omega$, and $\Gcal^\omega\vdash_E (a_n = \epsilon)$ by an $Eq_1$ application, giving $a_n =_{\Gcal^\omega} \epsilon$. If (2) holds, then $(a_n\not = \epsilon)\in\Gcal_{j+1}\subseteq \Gcal^\omega$.
\item If $(z\not = \epsilon)\in\Gcal^\omega$, then $(z\not = \epsilon)\in\Gcal_i$, for some $i\in\Ncal$. Then there exists $j > i$ s.t. $\phi(j) = (O,m,\Sbb,R,\{(z\not = \epsilon)\})$. Then $\Gcal_{j+1} = \Gcal_j\cup\{(a_{4j},a_{4j+1}\simp z),(a_{4j}\not = \epsilon),(a_{4j+1}\not = \epsilon)\}\subseteq\Gcal^\omega$.
\setcounter{enumi}{18}
\item If $(x,y\simp z)\in\Gcal^\omega$ and $(u,v\simp z')\in\Gcal^\omega$ and $z =_{\Gcal^\omega} z'$. There must be some $i\in\Ncal$ s.t. $\{(x,y\simp z),(u,v\simp z')\}\subseteq\Gcal_i$ and $\Gcal_i\vdash_E (z = z')$. Suppose $(x,y\simp z)$ and $(u,v\simp z')$ are distinct, then there exists $j > i$ s.t. $\phi(j) = (O,m,\CSbb,\{(x,y\simp z),(u,v\simp z')\},I)$, and $\{(x,y\simp z),(u,v\simp z')\}\subseteq\Gcal_j$, $\Gcal_j\vdash_E (z = z')$ hold. By construction, $\Gcal_{j+1} = \Gcal_j\cup\{(a_{4j},a_{4j+1}\simp x),(a_{4j},a_{4j+2}\simp u),(a_{4j+2},a_{4j+3}\simp y),(a_{4j+1},a_{4j+3}\simp v)\}\subseteq\Gcal^\omega$. If $(x,y\simp z)$ and $(u,v\simp z')$ are the same, a similar argument can be applied by using $\CSbb_C$. 
\end{enumerate}
The other cases are similar to the original proof. Note that the subscript of new labels needs to be adjusted accordingly.\qed
\end{proof}

Therefore the limit sequent in the new definition is indeed an Hintikka sequent, and we can extract an infinite counter-model from it. The completeness result follows.

\subsection{Proofs for Section~\ref{sec:experiments}}
\label{sec:app4}

\noindent Proof of Proposition~\ref{prop:admissible_a_c}.

\begin{proof}
We show that every derivation in $\lsbbi + C$ can be transformed
into one with no applications of $A_C.$ It is sufficient to show that we can
eliminate a single application of $A_C$; then we can eliminate all $A_C$
in a derivation successively starting from the topmost applications in 
that derivation. 
So suppose we have a derivation in $\lsbbi + C$ of the form: 
\begin{center}
\alwaysNoLine
\AxiomC{$\Pi$}
\UnaryInfC{$\myseq{(x,w\simp x);(y,y\simp w);(x,y\simp x);\Gcal}{\Gamma}{\Delta}$}
\alwaysSingleLine
\RightLabel{\tiny $A_C$}
\UnaryInfC{$\myseq{(x,y\simp x);\Gcal}{\Gamma}{\Delta}$}
\DisplayProof
\end{center}
where $w$ is a new label not in the root sequent. 
This is transformed into the following derivation: 
\begin{center}
\alwaysNoLine
\AxiomC{$\Pi'$}
\UnaryInfC{$\myseq{(x,\epsilon\simp x);(\epsilon,\epsilon\simp \epsilon);\Gcal[\epsilon/y]}{\Gamma[\epsilon/y]}{\Delta[\epsilon/y]}$}
\alwaysSingleLine
\RightLabel{\tiny $U$}
\UnaryInfC{$\myseq{(x,\epsilon\simp x);\Gcal[\epsilon/y]}{\Gamma[\epsilon/y]}{\Delta[\epsilon/y]}$}
\RightLabel{\tiny $C$}
\UnaryInfC{$\myseq{(x,\epsilon\simp x);(x,y\simp x);\Gcal}{\Gamma}{\Delta}$}
\RightLabel{\tiny $U$}
\UnaryInfC{$\myseq{(x,y\simp x);\Gcal}{\Gamma}{\Delta}$}
\DisplayProof
\end{center}
where $\Pi'$ is obtained by applying the substitutions $[\epsilon/y]$ and $[\epsilon/w]$
to $\Pi$ (by using Lemma~\ref{lem:subs}). Note that since $w$ does not occur in the root sequent,
$\Gcal[\epsilon/y][\epsilon/w] = \Gcal[\epsilon/y]$, 
$\Gamma[\epsilon/y][\epsilon/w] = \Gamma[\epsilon/y]$
and $\Delta[\epsilon/y][\epsilon/w] = \Delta[\epsilon/y].$
These substitutions do not introduce new instances of $A_C.$ \qed
\end{proof}\ \\

\noindent Proof of Lemma~\ref{lem:u_rule}.

\begin{proof}
The original $U$ rule can be separated into two cases: (1) $U'$, with the
restriction as described above and (2)
$U''$, where the created relational atom $(x,\epsilon\simp x)$
satisfies that $x$ does not occur in the conclusion. We show that 
case (2) is admissible, leaving case (1) complete.

A simple induction on the height $n$ of the derivation for
$\myseq{\Gcal}{\Gamma}{\Delta}$. Suppose $U''$ is the last rule in
the derivation, with a premise:
\begin{center}
$\myseq{(x,\epsilon\simp x);\Gcal}{\Gamma}{\Delta}$
\end{center}
where $x$ is a fresh label. Assume that the conclusion is not an empty
sequent, there must be some label $w$ that occurs in conclusion. By
Lemma~\ref{lem:subs}, replacing $x$ by $w$, we obtain that 
\begin{center}
$\myseq{(w,\epsilon\simp w);\Gcal}{\Gamma}{\Delta}$
\end{center}
is derivable in $n-1$ steps. By the induction hypothesis, there is a
$U''$-free derivation of this sequent, which leads to the derivation
of $\myseq{\Gcal}{\Gamma}{\Delta}$ by applying the restricted rule $U'$.\qed
\end{proof}